\documentclass[12pt,oneside,english]{article}
\usepackage{amsthm, amsmath, amssymb, euscript, mathrsfs, MnSymbol, verbatim, enumerate, multicol, multirow, color,babel, geometry, tikz,tikz-cd, tikz-3dplot, tkz-graph, array, hepth, enumitem, hyperref, thm-restate, thmtools, datetime}
\usepackage[utf8]{inputenc}
\usepackage[T1]{fontenc}
\usepackage[normalem]{ulem}

\numberwithin{equation}{section}
\numberwithin{figure}{section}

\usetikzlibrary{arrows,positioning,decorations.pathmorphing,   decorations.markings, matrix, patterns}

\hypersetup{colorlinks=true}
\hypersetup{linkcolor=black}
\hypersetup{citecolor=black}
\hypersetup{urlcolor=black}

\theoremstyle{plain}
\newtheorem*{thm*}{Theorem}
\newtheorem{thm}{Theorem}[section]

\newtheorem{lem}[thm]{Lemma}

\newtheorem*{cor}{Corollary}

\theoremstyle{definition}
\newtheorem{defn}[thm]{Definition}
\newtheorem*{defn*}{Definition}
\newtheorem{exmp}[thm]{Example}

\tikzset{
  big arrow/.style={
    decoration={markings,mark=at position 1 with {\arrow[scale=1.5,#1]{>}}},
    postaction={decorate},
    shorten >=0.4pt},
  big arrow/.default=black}

\makeatother

\begin{document}
\date{}
\institution{Northeastern}{\centerline{${}^{1}$Department of Mathematics, Northeastern University, Boston, MA, USA}}
\institution{HarvardPhys}{\centerline{${}^{2}$Department of Physics, Harvard University, Cambridge, MA, USA}}

\title{Characteristic numbers of crepant resolutions of Weierstrass models}
\authors{Mboyo Esole\worksat{\Northeastern}\footnote{e-mail: {\tt j.esole@northeastern.edu}} and Monica Jinwoo Kang\worksat{\HarvardPhys}\footnote{e-mail: {\tt jkang@physics.harvard.edu}} }

\abstract{
We compute characteristic numbers of crepant resolutions of Weierstrass models corresponding to elliptically fibered fourfolds  $Y$ dual in F-theory to a gauge theory with gauge group $G$. 
In contrast to the case of fivefolds, Chern and Pontryagin numbers of fourfolds  are invariant under crepant birational maps. It follows that Chern and Pontryagin numbers are independent on a choice of a crepant resolution.
We present the results for the Euler characteristic, the holomorphic genera, the Todd-genus, the $L$-genus, the $\hat{A}$-genus, and the curvature invariant $X_8$ that appears in M-theory. 
We also show that certain characteristic classes are independent on the choice of the Kodaria fiber characterizing the group $G$. That is the case of $\int_Y c_1^2 c_2$, the arithmetic genus, and the $\hat{A}$-genus. 
Thus, it is enough to know $\int_Y c_2^2$ and the Euler characteristic $\chi(Y)$ to determine all the Chern numbers of an elliptically fibered fourfold.    
We consider the cases of $G=$ SU($n$) for ($n=2,3,4,5,6,7$), USp($4$), Spin($7$), Spin($8$), Spin($10$),  G$_2$, F$_4$, E$_6$, E$_7$, or E$_8$. }
\maketitle

\tableofcontents
\newpage

\section{Introduction}

Characteristic classes are cohomology classes associated to isomorphic classes of  vector bundles \cite{Hirzebruch,Milnor,Fulton.Intersection}. They measure how a vector bundle is twisted or non-trivial.   Characteristic classes of a nonsingular variety are defined via its tangent bundle.

Characteristic classes are instrumental in many questions of geometry and theoretical physics. In string theory  and in  supergravity theories, characteristic classes appear in discussions of anomaly cancellations \cite{Gaume,AlvarezGaume:1983ig} and tadpole cancellations \cite{Sethi:1996es,CDE,Denef:2008wq}, in the computations of the index of supersymmetry operators, and in the definition of the charges of D-branes \cite{Minasian:1997mm,AE1} and orientifold planes \cite{Scrucca}.

The aim of this paper is to compute the characteristic numbers of elliptic fibrations  that are crepant resolutions of singular Weierstrass models given by the output of Tate's algorithm \cite{Tate}. Such elliptic fibrations are called  $G$-models.

\begin{defn}[Elliptic fibrations]
 A variety is said to be an elliptic $n$-fold if it is endowed with a proper surjective morphism $\varphi:Y\to B$ to a variety of dimension $n-1$  such that the generic fiber of $\varphi$ is a smooth projective curve of genus one, and $\varphi$ has a rational section. 
\end{defn}

 \begin{defn}[$G$-models]
 Let  $G$ be a  simple, simply-connected compact complex Lie group with Lie algebra $\mathfrak{g}$.  A $G$-model is an elliptic fibration $\varphi:Y\to B$ with a  discriminant locus  containing an irreducible component $S$ such  that 
 \begin{enumerate}
 \item the generic fiber over any other component of the discriminant is irreducible (that is, of Kodaira type I$_1$ or II), 
 \item the fiber over  the generic point of $S$ has a dual graph that becomes of the same type 
 as the Dynkin diagram of the Langlands dual of $\mathfrak{g}$ after removing the node corresponding to the component touching the section of the elliptic fibration. 
 \end{enumerate}
 \end{defn}

$G$-models are used to geometrically engineer gauge theories in compactifications of M-theory and F-theory \cite{Vafa:1996xn,Morrison:1996na,Bershadsky:1996nh,Katz:2011qp}.
They also play an important role in studying superconformal gauge theories (for a review see \cite{Heckman:2018jxk} and reference therein). 
 $G$-models are typically defined by crepant resolutions of singular Weierstrass models given by Tate's algorithm and characterizing  a specific (decorated) Kodaira fiber \cite{Andreas:2009uf,G2,Euler,EK1,SUG,EKY,EKY2,ES,ESY1,ESY2,EY,Lawrie,Mayrhofer:2014opa}.

We will restrict ourselves to the case of elliptic fibrations with a unique divisor over which the generic fiber is reducible. We also consider only elliptic fibration of complex dimension four with a trivial Mordell-Weil group. 
That corresponds to groups $G$ that are simple and simply-connected. We note that our techniques apply without subtleties to the cases with non-trivial Mordell--Weil groups or to cases where $G$ is semi-simple. 
The elliptic fibrations that we analyze are all given explicitly by crepant resolutions of singular Weierstrass models as discussed in section \ref{sec:GModels}.

Euler characteristics of $G$-models have been computed recently in \cite{Euler}. We continue the work started in \cite{Euler} by providing additional characteristic numbers  of $G$-models.
In our analysis,  we do not impose the Calabi-Yau condition since it will reduce the problem  to a computation of Euler characteristics, which is already addressed in \cite{Euler}. 
The Calabi-Yau fourfold case is reviewed  in section \ref{sec:CY4}.  We also point out integrality issues for the invariant $\chi_1(Y)$ in the  cases of SU($5$), SU($6$), and E$_6$ in section \ref{Sec:Oddities}. 
We  compute the characteristic numbers by explicitly performing pushforwards after correcting the Chern classes to take into account the effect of the sequence of blowups necessary to define the crepant resolution. 

The theory of characteristic classes was founded in the 1930s and 1940s by  Whitney,  Stiefel, Pontryagin, and Chern \cite{Hirzebruch,Milnor}. 
To these days,  the most famous characteristic numbers are the Stiefel-Whitney numbers, the Pontryagin numbers, the Chern numbers, together with the Euler characteristic, which is the oldest topological invariant. 
 The theory of characteristic classes relies deeply on sheaf theory as developed by  Kodaira, Spencer, and Serre; and historically, also  on Thom's theory of cobordism.  
Grothendieck introduced an axiomatic definition of Chern classes in the Chow ring of a variety using projective bundles.   
This establishes characteristic classes as familiar objects in intersection theory \cite[Chap 3]{Fulton.Intersection}. 
In a sense, any natural transformation from the complex vector bundles to the cohomology ring is a polynomial in the Chern numbers. In his seminal book, Hirzebruch expresses all characteristic classes in terms of Chern classes \cite{Hirzebruch}.

\subsection{List of characteristic numbers}

We compute  the following six types of rational Chern and Pontryagin numbers for each $G$-models:
\begin{enumerate}
\item The Chern numbers 
\begin{equation}
\text{$\int_Y c_1(TY)^4$, $\int_Y c_1(TY)^2 c_2 (TY)$, $\int_Y c_1(TY) c_3 (TY)$, $\int_Yc_2^2(TY)$,  and $\int_Y c_4 (TY).$}
\end{equation}
\item The holomorphic genera  $\chi_p(Y) = \sum_{q=0}^n (-1)^q h^{p, q}(Y)$ \cite{Klemm:1996ts}:
  \begin{equation}
  \begin{aligned}
  \chi_0(Y)&=\int_Y \mathrm{Td}(TY)=\frac{1}{720}\int_Y (-c_4 + c_1 c_3 +3 c_2^2 + 4 c_1^2 c_2 -c_1^4),\\
  \chi_1(Y) & = \frac{1}{180}\int_Y (-31c_4 -14 c_1 c_3 +3 c_2^2 + 4 c_1^2 c_2 -c_1^4),\\
  \chi_2(Y) & = \frac{1}{120}\int_Y (79c_4 -19 c_1 c_3 +3 c_2^2 + 4 c_1^2 c_2 -c_1^4).
  \end{aligned}
  \end{equation}
  The holomorphic Euler characteristic $\chi_0$ is a birational   \cite[Theorem II.8.19]{Hartshorne} and diffeomorphism invariant. 
  The other holomorphic genera are diffeomorphism invariants and invariant under crepant birational maps. 

\item The  Pontryagin numbers of a fourfold are the numbers $\int_Y p_2(TY)$ and $\int_Y p_1^2(TY)$, where the Pontryagin classes $p_1(TY)$ and $p_2(TY)$ are defined as follows:
 \begin{equation}
 \begin{aligned}
 p_1(TY) &=c_1^2(TY) -2c_2(TY), \\
  p_2(TY) &= c_2^2(TY)-2 c_1(TY) c_3(TY)+2c_4(TY).
 \end{aligned}
 \end{equation}
\item The Hirzebruch signature of a fourfold, 
 \begin{equation}
\sigma(Y)=\frac{1}{45}\int_Y \Big(7p_2(TY) -p_1^2(TY)\Big)=\frac{1}{45}\int_Y (-c_1^4+4 c_1^2 c_2+3 c_2^2-14 c_1 c_3+14 c_4).
\end{equation}
The signature is the degree of the Hirzebruch $L$-genus. The $L$-genus  is always an integer  \cite[Corrolary 19.5, p. 226]{Milnor} and depends only on the oriented homotopy type of the variety \cite[Corrolary 19.6, p. 226]{Milnor}.
\item The $\hat{\text{A}}$-genus of a fourfold,
\begin{equation}
\begin{aligned}
\int_Y \hat{\text{A}}_2(TY)&=\frac{1}{5760}\int_Y\Big(7p_1^2(TY)-4p_2(TY)\Big)\\
&=\frac{1}{5760}\int_Y\Big( {7 c_1^4-28 c_1^2 c_2+8 c_1 c_3+24 c_2^2-8 c_4} \Big).
\end{aligned}
 \end{equation}
By the Atiyah-Singer theorem, if the fourfold $Y$ is a spin manifold,  the degree of  $ \hat{\text{A}}_2$ gives the index of the Dirac operator on  $Y$. We will see that the $\hat{\text{A}}$-genus is independent of our choice of a crepant resolution and is also independent of $G$. 
\item 
 We also compute the following form that plays an important role in many questions of anomaly cancellations, it  detects the appearance of a non-vanishing contribution to the one-point function for the two, three, and four forms in type IIA, M, and F-theory \cite{Vafa:1995fj,Sethi:1996es}:
\begin{equation}
X_8(Y)=\frac{1}{192 }\int_Y \Big( p_1^2(TY) -4p_2(TY)\Big).
\end{equation}
In string theory, $X_8(Y)$ typically appears as a curvature invariant. As it is  expressed by Pontryagin numbers, $X_8(Y)$ is an oriented  diffeomorphism invariants and is also invariants under crepant birational maps. 
 Our computations of $X_8(Y)$ has some overlaps with  \cite{Lawrie:2018jut}, see Appendix \ref{AD} for more information. 
\end{enumerate}

\subsection{Chern numbers of fourfolds are $K$-equivalence invariants}

Crepant resolutions are in a sense the mildest form of desingularizations, since they do not modify a singular  variety away from its singular locus and  preserves the canonical class. 
Crepant resolutions are also relative minimal models over the underlying singular variety. 
When they exist, crepant resolutions are not necessarily unique for varieties of dimension three or higher. For $G$-models, the number of flops can be pretty big \cite{EJJN1,EJJN2,Hayashi:2014kca}. 
When two varieties are crepant resolutions of the same underlying singular variety, a natural question to ask is if their characteristic numbers are the same. 
 For instance, the Betti numbers and  the Hodge numbers are invariants under crepant birational maps \cite{Batyrev.Betti,Kontsevich.Orsay}.

A key point that makes this paper possible is the fact that  it is enough to know a single crepant resolution to compute the Chern and Pontryagin numbers of a given $G$-model. 
This is because the Chern numbers of a fourfold are invariants under crepant birational maps as proven in Theorem \ref{Thm:TheInvariance} using  a result of  Aluffi \cite[p. 3368]{Aluffi.IMRN}  and the birational invariance of the Todd genus.  
We would like to point out that such an invariance for the Chern numbers should not be taken for granted as it is not generally true that Chern and Pontryagin numbers are invariant under crepant birational maps. 
The first counter-example appears in dimension five (see Example \ref{GoMAEx1}). This fact motivated our choice to present the results only for fourfolds.
For a projective five-fold, the Chern numbers are 
$$
\int c_1^5, \quad \int c_1^3 c_2 , \quad \int c_1^2 c_3, \quad \int c_1 c_4, \quad \int c_5, \quad \int c_1 c_2^2, \quad \int c_2 c_3,
$$
where only the first five are invariant under K-equivalence  by Aluffi's theorem (see Theorems \ref{Thm:Al} and \ref{Thm:TheInvariance}). In fact, 
Goresky and MacPherson gave the following example of a five-dimensional Schubert
variety with two different small resolutions with  the same Chern
numbers with the exception of $\int c_2 c_3$.  For more information on $K$-invariance, see \cite{Wang03,Wang2002}.

\begin{exmp}[Goresky and MacPherson, {\cite[Example 2, page 221]{GoMa}}]\label{GoMAEx1}
Let $X$ be the Schubert variety in the Grassmannian $G_2(\mathbb{C}^4)$ consisting of all complex two-planes $V\subset \mathbb{C}^4$ such that $\dim( V \cap\mathbb{C}^2)\geq 1$. 
This variety $X$ has a singularity at the point $V=\mathbb{C}^2$ , and it has a small resolution $\varphi_1: \widetilde{X}_1\to X$ where $\widetilde{X}_1$ consists of all  $(1,2)$- flags $V^1\subset V^2\subset\mathbb{C}^4$  such that $V^1\subset\mathbb{C}^2$. 
It has a second small resolution $\varphi_2: \widetilde{X}_1\to X$ which consists of all $(2,3)$ flags $V^2\subset V^3\subset\mathbb{C}^4$  such that $\mathbb{C}^2\subset V^3$. Although $\widetilde{X}_1$ and $\widetilde{X}_2$ are homeomorphic, by a computation of Verdier, they do not have the same Chern classes.  
\end{exmp}

\begin{exmp}[Goresky and MacPherson, {\cite[Example 2, page 222]{GoMa}}]
Let $X$ be a Schubert variety $X=\{V\in G_2(\mathbb{C}^5|\dim(V\cap \mathbb{C}^3)\leq 1\}$.  
Let $\widetilde{X}_1$ be the variety of partial flags $V_1\subset V_2\subset \mathbb{C}^5$ 
such that $V_1\subset \mathbb{C}^2$. Let $\tilde{X}_2$ be the variety of 
partial flags $V_2\subset V_4\subset\mathbb{C}^4$ such that 
$\mathbb{C}^3\subset V_4$. Both $\widetilde{X}_1$ and $\widetilde{X}_2$ are small resolutions 
of $X$ but their cohomology rings are not even abstractly isomorphic. 
\end{exmp}

The following theorem of Libgober and Wood was known to Hirzebruch in the case of fourfolds.
\begin{thm}[Libgober and Wood, {\cite{LW}, Theorem 3}]
For a compact complex manifold $X$, the Chern number $\int_X c_1c_{n-1} $ is determined by the holomorphic genera and hence by the Hodge numbers. 
\end{thm}

Even though Chern numbers other than the Euler characteristic are not topological invariants, some are invariant under crepant birational maps as proven by Aluffi.

\begin{thm}[Aluffi, { \cite[page 3368]{Aluffi.IMRN}}]\label{Thm:Al}
For two nonsingular $n$-dimensional complete 	varieties $X$ and $Y$ connected by a crepant birational map
$$
\int_X c_1(TX)^i c_{n-i}(TX)=\int_Y c_1(TY)^i c_{n-i}(TY), \quad i=0, 1, \dots, n.
$$
\end{thm}
The following theorem is also used in this paper.
\begin{thm}[Aluffi, { \cite[Corollary 1.2]{Aluffi.IMRN}}]
Let $f:\widetilde{X}\to X$ be a crepant resolution. Then the class $$f_*\Big( c(T\widetilde{X})\cap[\widetilde{X}]\Big),$$
 in $(A_* X)_{\mathbb{Q}}$ is independent of $X$. 
\end{thm}

The following theorem assets that the Chern numbers of fourfolds are  $K$-equivalence  invariants.  The proof is a simple application of  Theorem \ref{Thm:Al} of Aluffi. 
\begin{thm}\label{Thm:TheInvariance}
The Chern and Pontryagin numbers of an algebraic variety of complex dimension four are $K$-equivalence invariants.
\end{thm}
\begin{proof}
The Chern numbers of a fourfold are $\int_Y c_1(TY)^4$, $\int_Y  c_1(TY)^2 c_2 (TY)$, $\int_Y  c_1(TY) c_3 (TY)$, $\int_Y c_2^2(TY)$,  and $\int_Y c_4 (TY)$.  
Hirzebruch showed that  $\int_Y  c_1(TY)c_3(TY)$ can be expressed by Hodge numbers. Moreover, from a result of Aluffi, we know that the Chern numbers $\int_Y  c_1^i(TY)c_{n-i}(TY)$ are $K$-equivalence invariants.  The Chern number  $\int_Y  c_2^2(TY)$ is also a $K$-equivalence invariant as it can be expressed as a linear combination of the holomorphic Euler characteristic (which is a birational invariant) and the $\int_Y  c_1^i(TY)c_{n-i}(TY)$ for $i=0,1,2,3,4$. It follows that all Chern numbers of a fourfold are $K$-equivalence invariants.  The same is true for Pontryagin numbers since they are linear combinations of Chern numbers.
\end{proof}

 Theorem \ref{Thm:TheInvariance} allows us to compute Chern numbers of $G$-models in the crepant resolution of our choice since  it is independent of a choice of a crepant resolution.

The Chern number $\int_Yc_1(TY)^2 c_2(TY)$, the $\hat{A}$-genus, and the Todd-genus (the holomorphic Euler characteristic) are invariants of the choice of $G$. 
They can all be expressed as invariants of the divisor $W$ defined by the vanishing locus of a smooth section of $\mathscr{L}$. 

\subsection{Independence of Chern numbers}
By expressing $\chi_0$ and $\hat{A}$ in terms of Chern numbers and using the identity $\int_Y c_1^4=0$, which holds for any crepant resolution of a Weierstrass model (see Theorem \ref{thm.1.4}), we get the following expressions of $c_1^2 c_2$ and $c_1 c_3$:
\begin{align}
\int_Y c_1^2c_2=96(\chi_0(Y)-\hat{A}(Y)),\quad 
\int_Y c_1 c_3 = 384\hat{A}(Y)+ 336\chi_0(Y)+\chi(Y)- 3 \int_Y c_2^2.
\end{align}
For an elliptic fibration that is a crepant resolution of a Weierstrass model, this shows that $\int_Y c_1^2 c_2$ gives the same value as a smooth Weierstrass model with the same  fundamental line bundle $\mathscr{L}$. 
It is therefore enough to compute only $\int_Y c_2^2$ and the Euler characteristic $\chi(Y)=\int_Y c_4$ to know all the Chern numbers of a crepant resolution of a Weierstrass model.

\section{Strategy and roadmaps of  results}

The data at the heart of our computations are the  lists of blowups that give a  crepant resolution for each of these $G$-models. 
We can summarize our strategy as follows 
\begin{enumerate}\item 
The  $G$-models considered in this paper are defined by crepant resolutions of Weierstrass models given  in section \ref{sec:Tate}. 
\item They crepant resolutions are obtained by  sequence of blowups given in Table  \ref{tab:blowupcenters}. 
\item The Chern numbers and Pontryagin numbers of the $G$-models are listed respectively in Table 
\ref{Table.ChernNumbers} and  Table \ref{Table.Pontryagin}. 
\item The holomorphic genera are listed in Table \ref{Table.HolomorphicEC}. 
\item The invariant $X_8(Y)$, the signature $\sigma(Y)$, and the $\mathnormal{\hat{A}}$-genus are given in Table \ref{Table.Pontryagin2}. 
 \end{enumerate}
   Let  $\pi:X_0\to B$ be the projective bundle in which the Weierstrass model $Y_0$ is defined.  
   We define a crepant resolution $f:Y\to Y_0$ by  a composition of $k$ blowups   $f_i: X_{i}\to X_{i-1}$  $i=1,\cdots, k$ with smooth centers  that are transverse complete  intersections in $X_i$. 
  Using Theorem \ref{Thm:AluffiCBU} for each $f_i$, we  compute the Chern classes  of each $X_i$ and by adjunction. We then determine the Chern classes of $Y$ by adjunction. 
 
For  an element $Q$ of the Chow ring $A^*(Y)$, we compute $\int_Y Q$ as a function of the topology of the base as follows. 
Using Theorem  \ref{Thm:Push}, we then express the Chern numbers in terms of the Chow ring of the original ambient space $X_0$ in which the Weierstrass model is defined. 
Next, we pushforward these to the base $B$ of the elliptic fibration using Theorem \ref{Thm:PushH}. Since we fix the base to be a threefold, we can also simply use Lemmas \ref{lem:Push2}--\ref{lem:PushH} for our pushforwards. In summary we get,
\begin{equation}\int_Y Q=\int_{X_k} Q\cdot [Y]=\int_{B} \pi_*  f_{1*} f_{2*} \cdots f_{k*} \Big(Q\cdot [Y]\Big).
\end{equation}

The following theorem gives the behaviors of intersection numbers involving Chern classes and Pontryagin classes of dimension too small to give Chern or Pontryagin numbers. 
To give a number, they must be multiplied by an element of the Chow ring of appropriate dimension. 
\begin{thm}\label{thm.1.4}
Let  $\varphi:Y\to B$ be an  elliptic fibration
given by the crepant resolution of a singular Weierstrass model of dimension $n$. Then,
\begin{align}
& \int_Y c_1^i(TY)\cdot  \alpha &&=\int_B (c_1-L)^i\cdot \varphi_* \alpha, \quad &&\alpha \in A^*(Y),\label{Thm21.1}  \\
& \int_Y c_1^i(TY)\cdot  \varphi^* \beta &&=\int_B (c_1-L)^i \beta,  &&\beta \in A^*(B), \label{Thm21.2} \\
& \int_Y c_1^n(TY) &&=\int_B (c_1-L)^n=0, \label{Thm21.3}\\
& \int_Y c_3(TY)\cdot  \varphi^* \beta  &&=\int_B \varphi_* (c_3 (TY) [Y]) \cdot \beta, \quad &&\beta \in A^*(B), \label{Thm21.4}\\
& \int_Y c_2(TY)\cdot  \varphi^* \beta  &&=12\int_B L \cdot \beta, \quad &&\beta \in A^*(B), \label{Thm21.5}\\
& \int_Y p_1(TY)\cdot  \varphi^* \beta  &&=
 \int_B(c_1-L)^2 \beta-24\int_B L\cdot \beta \quad &&\beta \in A^*(B)\label{Thm21.6}
\end{align}
where $ \varphi_* (c_3 (TY) [Y])$ is given in Table \ref{Table:ChernNumbers} for the $G$-models considered in the paper and is invariant of a choice of a crepant resolution. 
\end{thm}

\begin{proof}
Since $Y$ is a crepant resolution of a Weierstrass model, we have that the first Chern class of $Y$ is the pullback $c_1(TY)=\varphi^*(c_1-L)$.
Equation \eqref{Thm21.1} is therefore a direct consequence of the projection formula and the invariance of the degree under a proper map: 
$$
\int_Y \varphi^* (c_1-L)^i \alpha =\int_B \varphi_* (\varphi^* (c_1-L)^i \alpha)= \int_B (c_1-L)^i \varphi_*\alpha.
$$ 
Equations \eqref{Thm21.2} and \eqref{Thm21.3} are direct specializations of equation \eqref{Thm21.1}. 
In particular, if $Y$ is an $n$-fold, we have $\int_Y \varphi^* (c_1-L)^n=\int_B (c_1-L)^n=0$. 
 Equation \eqref{Thm21.4} is also a direct consequence of the projection formula.
 Equation \eqref{Thm21.5} follows from Theorem \ref{Thm.C2} and the fact that $\int_{Y_0}\pi_* c_2 = 12L$ for a smooth Weierstrass model $\pi: Y_0\to B$. 
  Theorem \ref{Thm.C2}, asserts that  for any crepant blowup centered at a transverse complete intersection of  smooth divisors, we have $f_*(c_2(TY)\cdot[Y])=c_2(TY_0)\cdot [Y_0]$. 
 We note that  these are exactly the types of blowups we use in  section \ref{sec:GModels}. 
 Equation \eqref{Thm21.6}  is derived by linear combination using $p_1=c_1^2-2c_2$ and equations \eqref{Thm21.2} and \eqref{Thm21.5}.
  \end{proof}
  
  \begin{table}[htb]
\begin{center}
\renewcommand{\arraystretch}{1.8}
\scalebox{1}{$
\begin{array}{|c|c|}
\hline 
\text{Algebra} & \mu_G =\int_Yc_2^2(TY)-\int_{Y_0}c_2^2(TY_0) \\\hline
\text{A}_1 & -2 \int_BS (7 L-S)^2 \\
\hline
\text{A}_2, \ \text{C}_2, \ \text{G}_2 &-8\int_B S \left(19 L^2-8 L S+S^2\right) \\
\hline
\text{A}_3, \ \text{B}_3 & -4 \int_BS \left(50 L^2-28 L S+5 S^2\right) \\
\hline
\text{D}_4, \ \text{F}_4 & -8 \int_BS \left(27 L^2-16 L S+3 S^2\right) \\
\hline 
\text{A}_4 & -5 \int_BS \left(50 L^2-35 L S+8 S^2\right) \\
\hline 
\text{D}_5 & -4 \int_BS \left(63 L^2-44 L S+10 S^2\right) \\
\hline 
\text{A}_5 & - \int_BS \left(298 L^2-251 L S+70\right) \\
\hline 
\text{A}_6 & -2 \int_BS \left(174 L^2-171 L S+56 S^2\right) \\
\hline 
\text{E}_6 & -3 \int_BS \left(86 L^2-61 L S+14 S^2\right) \\
\hline 
\text{E}_7 &\int_B \left(135 L^2-100 L S+24 S^2\right) \\
\hline 
\text{E}_8 &\int_B \left(8 L^2-7 L S+2 S^2\right) \\
\hline 
\end{array}
$}
\end{center}
\caption{$\mu_G$ for all the G-models. $Y$ is the crepant resolution of a singular Weierstrass model corresponding to a $G$-model, and $Y_0$ is a smooth Weierstrass model over the same base with the same fundamental line bundle $\mathscr{L}$. \label{Table.muG} }
\end{table}

  \begin{cor}\label{Cor.c2}
  Let  $\varphi:Y\to B$ be an  elliptic fibration
given by the crepant resolution of a singular Weierstrass model of dimension $n$ with fundamental line bundle $\mathscr{L}$. Then,
\begin{equation}
  \int_Y c_1(TY)^{n-1}c_2(TY) = 12\int_B (c_1-L)^{n-1} L.
 \end{equation}
  \end{cor}
The following theorem shows how the dependence of the Pontryagin numbers on the Kodaira type (or better, on the group $G$) is controlled by the Chern number $\int_Y c_2^2(TY)$. 

\begin{thm}\label{thmP}
For a $G$-model,
defined by  a crepant resolution of one of the  Weierstrass model given  in section \ref{sec:Tate} and resolved by the sequence of blowups given in Table  \ref{tab:blowupcenters}, we have 
\begin{equation}
\begin{aligned}
\int_Y c_2^2(TY) &=\    \   24\int_B L (c_2-c_1 L+6 L^2)+\mu_G, \\
\int_Y p_2 (TY)&=-24 \int_B L (2 c_2-c_1^2+36 L^2)+7\mu_G,\\
\int_Y  p_1^2(TY)&=\  \   48\int_B L (2 c_2-c_1^2+11 L^2)+4\mu_G,
 \end{aligned}
\end{equation}
where $\mu_G= \int_B S (\alpha_0 L^2+ \alpha_1 LS+ \alpha_2 S^2)$ is the contribution from the singularities induced by the Kodaira type over $S$.  The different values of $\mu_G$ are listed in Table \ref{Table.muG}.
The correction for $\lambda_2\int_Y p_2 + \lambda_1 \int_Y p_1^2$ is then $(7\lambda_2+ 4\lambda_1)\mu_G$.
In particular, 
\begin{align}
\begin{aligned}\label{Eq.p}
\int_Y\Big( p_1^2(TY)-4p_2(TY)\Big) &=48\int_B L \left(c_1^2-2 c_2-61 L^2\right)-24 \mu_G, \\
\int_Y\Big(7p_2(TY)-p_1^2(TY)\Big) &=120\int_B L \left(2 c_2-c_1^2+46 L^2\right)+45\mu_G,\\
\int_Y\Big(7p_1^2(TY)-4p_2(TY)\Big) & = 240 \int_BL \left(2 c_2-c_1^2+L^2\right).
\end{aligned}
\end{align}
\end{thm}
\begin{proof}
 This follows directly from an inspection of Table \ref{Table.ChernNumbers},  Table \ref{Table.Pontryagin}, and Table \ref{Table.Pontryagin2}. 
\end{proof}

\section{Geometric discussion}
\subsection{The Calabi-Yau fourfold case}\label{sec:CY4}
In the case of Calabi-Yau fourfolds, knowing the Euler characteristic is enough to also compute other invariants such as the Chern number $c_2(TY)^2$ is a function of the Euler characteristic. 
\begin{thm}\label{Thm2.CY4}
 The Chern numbers and Pontryagin numbers of a  Calabi-Yau fourfold are topological invariants depending only on its Euler characteristic. 
\end{thm}
\begin{proof}
The only non-zero Chern numbers of a Calabi-Yau  fourfold are $\int_Y c_2^2(TY)$ and $\int_Y c_4(TY)=\chi(Y)$. These two Chern numbers  are related linearly as  
$$\int_Y c_2^2(TY)=480+\frac{1}{3}\chi(Y).$$ Thus, for a Calabi-Yau fourfold, all Chern numbers and Pontryagin numbers are topological invariants as they are functions of the Euler characteristic of $Y$. 
\end{proof}
 \noindent For more information see section \ref{sec:CY4}. 
It follows that in the Calabi-Yau fourfold case, all Chern and Pontryagin numbers  can be read from the computation of the Euler characteristic of $G$-models in \cite{Euler} as  we have \cite{Klemm:1996ts,Sethi:1996es}
  $$
  \begin{aligned}
\int_Y  c_2^2(TY) &=480+\frac{1}{3}\chi(Y),\quad
  \sigma =32+\frac{1}{3}\chi(Y),\\
  \chi_0&=2 ,\quad
  \chi_1  =8-\frac{1}{6} \chi(Y),\quad
  \chi_2  =12+ \frac{2}{3}\chi(Y),\\
  X_8&=-\frac{1}{24}\chi(Y),\quad
 \frac{1}{5760} \int_Y \hat{\text{A}}_2 = 2.
   \end{aligned}
  $$
For $G$-models, the Euler characteristic of a Calabi-Yau fourfold is given in Table 10 of \cite{Euler}, which we reproduce here for completeness. We notice that the same table can be obtained from the Euler characteristic $c_4(TY)$ after putting $L=c_1$, which is the condition that ensure that the canonical class of $Y$ is trivial. 

\begin{table}[htb]
\begin{center}
\scalebox{.98}{$ 
 \begin{array}{|c|c |}
 \hline
 \text{Models}  & \text{ $\chi(Y_4)$,   Euler characteristic} \\
 \hline
\text{Smooth Weierstrass} &  12 \int_B(c_1 c_2 + 30 c_1^3) \\
 \text{SU}(2) & 6 \int_B (2 c_1 c_2 + 60 c_1^3 - 49 c_1^2 S + 14 c_1 S^2 - S^3)\\
 \text{SU}(3)\ \text{or}\   \text{USp}(4) \ \text{or}\  G_2 & 12 \int_B (c_1 c_2 + 30 c_1^3 - 38 c_1^2 S + 16 c_1 S^2 - 2 S^3)\\
\text{SU}(4)\ \text{or}\  \text{Spin}(7)&  12 \int_B(3 c_1 c_2 + 30 c_1^3 - 50 c_1^2 S + 28 c_1 S^2 - 5 S^3))\\
\text{Spin}(8) \ \text{or} \  F_4& 12\int_B  (c_1 c_2 +30 c_1^3 - 54 c_1^2 S + 32 c_1 S^2 - 6 S^3) \\
\text{SU}(5)&  3\int_B (4 c_1 c_2+120 c_1^3-250 c_1^2 S+175 c_1 S^2-40 S^3)\\
\text{SU}(6)& 3\int_B (4c_1 c_2+120c_1^3 -298 c_1^2S+251 c_1 S^2-70 S^3)\\
\text{SU}(7)& 6\int_B (2c_1 c_2+60c_1^3-174 c_1^2S+171 c_1 S^2-56 S^3)\\
\text{Spin}(10)& 12\int_B (c_1c_2+30c_1^3-63 c_1^2 S+44 c_1 S^2-10 S^3) \\
 \text{E}_6& 3\int_B  (4 c_1 c_2+120 c_1^3  - 258 c_1^2 S + 183 c_1 S^2 - 42 S^3)\\
 \text{E}_7& 6\int_B ( 2 c_1 c_2+60 c_1^3 - 135 c_1^2 S + 100 c_1 S^2 - 24 S^3)\\
\text{E}_8&  12\int_B  ( c_1 c_2+ 30 c_1^3 - 80 c_1^2 S + 70 c_1 S^2 - 20 S^3)\\
\hline
\end{array}
$}
\caption{Euler characteristic for Calabi-Yau elliptic fourfolds where $c_1=L$.}
\label{table:CY4}
\end{center}
\end{table}

\subsection{A geometric interpretation of the $\hat{A}$-genus and the Todd-genus of  a $G$-model. }

In this section, we will see that the $\hat{A}$-genus of a $G$-model $Y$ can be understood as the $A$-genus of the surface $W$ with normal bundle $\mathscr{L}$ in the base $B$. 
We call $W$ the {\em Weierstrass divisor}.

The $\hat{\text{A}}$-genus of a $G$-model is 
$$
\int_Y \hat{A}(Y)=\frac{1}{5760}\int_Y(7p_1^2-4p_2)= \frac{1}{24}\int_B L(-c_1^2+2c_2+L^2)=\int_{W} \hat{A}_1(W).
$$
We compute the Chern classes of $W$ by the adjunction formula: 
$$
c(TW)=\frac{c(TB)}{1+L}=1+(c_1-L)+(c_2-c_1 L +L^2).
$$
It is then direct to determine the first Pontryagin class of $W$:
$$
p_1(TW)=c_1(TW)^2-2c_2 (TW)= (c_1-L)^2 -2 (c_2-c_1 L +L^2)=c_1^2-2c_2-L^2.
$$
Since $W$ is a surface, it has a unique  Pontryagin number, namely $\int_W p_1(TW)$:
$$
\int_W p_1(TW)=\int_B L(c_1^2-2c_2-L^2)=-24 \int_W \hat{A}(W).
$$
The last equality uses the fact that  $\hat{A}_1=-p_1/24$.
If $W$ is a spin manifold, $\int_{W} \hat{A}_1(W)$ is an integer number and so is $\hat{A}(Y)$. 
We  note that $Y$   has a vanishing canonical class if and only if  $W$ does too. In other words, $Y$ is  Calabi-Yau if and only if  $W$ is a K3 surface. 
In that case, $\hat{A}(Y)=\hat{A}(W)=2$.

 A direct computation shows that the arithmetic genus of $Y$ is the same as the arithmetic genus of the surface $W$.
\begin{align}
\chi(W, \mathscr{O}_W)&=\int_B (1-e^{-L}) \mathrm{Td}(TB)=
\frac{1}{12}\int_W\Big( c_1(TW)^2+c_2(TW)\Big)\\
&=\frac{1}{12}\int_B L(c_1^2-3 c_1 L+c_2+2 L^2)=\chi_0(Y).
\end{align}

\subsection{Integrability problems for SU($5$) and E$_6$ Calabi-Yau fourfolds.}\label{Sec:Oddities}
A Calabi-Yau fourfold $Y$ has a  trivial canonical class and therefore satisfies the following relation 
\begin{equation}
\chi_1(Y):=\chi(\Omega^1_Y)=\sum_{p=0}^4 (-)^p h^{1,p}(Y)=8-\frac{1}{6}\chi(Y).
\end{equation}
 Since $\chi_1(Y)$ is obtained by adding and subtracting Hodge numbers, it is an integer number. Thus, the Euler characteristic $\chi(Y)$ of a Calabi-Yau fourfold  should be a multiple of $6$. 
By a direct inspection of Table \ref{table:CY4}, we observe that this is indeed the case for each $G$-models without any additional condition with the exception of SU($5$), SU($6$), and E$_6$. 
\begin{lem} Let $S$ be the smooth divisor supporting a generic fiber of  Kodaira type I$_5^s$, I$_6^s$ or IV$^{* \text{ns}}$ in the \textnormal{SU($5$)}, \textnormal{SU($6$)} or \textnormal{E$_6$} model.  The 
The Euler characteristics of the \textnormal{SU($5$)}-model , \textnormal{SU($6$)}, and the \textnormal{E$_6$}-model are  divisible by $6$ in the Calabi-Yau case if and only if  
\begin{equation}
\int_B c_1(TB) S^2\quad \text{is an even integer number}.
\end{equation}
\end{lem}
This condition is clearly not always realized as seen in Example  \ref{Examp.no}. If $c_1(TB)$ or $S$ is an even class, or if $B$ is a flat fibration with fiber $S$,  the condition is automatically satisfied. 
A  geometric situation naturally realizing the condition that $\int_B c_1(TB) S^2$ is even is presented in Example \ref{Examp.Curve}.
\begin{exmp}\label{Examp.no}
If $B$ is a smooth quadric threefold in $\mathbb{P}^4$ and $S$ is the intersection of  a hyperplane of $\mathbb{P}^4$ with $B$, then,  $\int_B c_1(TB) S^2=(5-2) 1^2=3$. 
Note that a smooth quadric threefold is Fano. Thus we can construct a Weierstrass model with $\mathscr{L}$ the anticanonical line bundle of $B$. 
\end{exmp}

\begin{exmp}\label{Examp.Curve}
 The line bundle $\mathscr{O}_B(S)$ has two distinct sections whose zero loci are smooth transverse divisors in $B$, then we can compute their complete intersection, which gives a smooth curve in $B$. 
By the adjunction formula, the first Chern class of such a curve would be 
$
c_1(TB)-2 S$. 
It follows that its Euler characteristic is 
$
2-2g=\int_B ( c_1 -2S)S^2.
$  
In such a case, we see that $\int_S c_1 S^2=2(1-g+\int_B S^3)$ is an even number. 
\end{exmp}

\section{Invariance of characteristic numbers}

  In contrast to Chern numbers, Pontryagin numbers are oriented diffeomorphism invariants. 
  There are even oriented homeomorphic invariants as proven by Novikov in the 1960s \cite{Novikov}. 
It follows that characteristic classes that are expressed in terms of Pontryagin classes and the Euler characteristic are both oriented diffeomorphic invariants and oriented homeomorphic invariants. 
   In contrast to the Euler characteristic,  Pontryagin numbers change their signs when the variety changes its orientation. 
Thus, if any Pontryagin class of a variety is non-zero, the variety cannot possess any orientation reversing diffeomorphism (or homeomorphism).

In 1954, Hirzebruch asked in \cite{Hirzebruch.Q} which linear combinations of Chern numbers 
of a nonsingular projective variety are topological invariants. 
Since  a complex projective 
variety comes with a special choice of orientation, it is natural to restrict Hirzebruch question  to oriented homeomorphism or an oriented diffeomorphism. 
With that restriction, the question was answered by  Kotschick in the following theorem.

\begin{thm}[Kotschick, {\cite{Kotschick1}}]
\begin{enumerate}
\item A rational linear combination of Chern numbers is an
oriented diffeomorphism invariant of a smooth complex projective
variety if and only if it is a linear combination of the Euler and
Pontryagin numbers.
\item In complex dimension $n\geq 3$, a rational linear combination of Chern numbers is a
diffeomorphism invariant of smooth complex projective varieties if and only if it is a multiple of
the Euler characteristic. In complex dimension two, both Chern numbers $\int c_1^2$ and $\int c_2$ are oriented diffeomorphism invariant.  
\end{enumerate}
\end{thm}
The statement about surfaces is a consequence of Seiberg-Witten's theory \cite{Kotschick2,Kotschick3}. 
 
There are also invariants such as the Euler characteristic, the holomorphic genus, the holomorphic genera, and the signature, which are linear combination of Hodge numbers. 
Hodge numbers are known to be the same for varieties in the same $K$-equivalence class. The following theorem characterizes which linear combinations of Chern numbers are also  linear combinations  of Hodge numbers.

\begin{lem}[Kotschick, {\cite{Kotschick1}}]
A rational linear combination of Chern numbers of a smooth complex projective variety
is determined by the Hodge numbers if and only if it is a linear combination of the holomorphic genera  
 $\chi_p$.
\end{lem}

The Euler characteristic is a homotopy invariant while the signature is an oriented homotopy invariant. The following theorem of Khan characterizes linear combination of Chern numbers that are oriented homotopy invariants. 

\begin{thm}[Khan, {\cite{Khan}}]
A linear rational combination of Chern numbers is an oriented homotopy invariant, for almost-complex manifolds, if and only if it is a rational linear combination of the Euler characteristic and the signature of the manifold. 
\end{thm}

\section{Pushforward of blowups and projective bundles}
\label{sec:push}
\subsection{A brief history of pushforward in string theory}
Intersection theory in algebraic geometry as we know it today is presented in the seminal book of Fulton \cite{Fulton.Intersection}. 
Chern classes are defined from Segre classes, which satisfy stronger functorial properties. 
The first formula computing  the  Chern classes  of the tangent bundle of the blow-up of
a nonsingular variety along a nonsingular center  was conjectured by Todd and Segre, and proven by  Porteous \cite{Porteous}. A generalization of Porteous to singular varieties was obtained by Aluffi   \cite{Aluffi_CBU}, who also provided user friendly descriptions of Porteous formula in the case of blow-ups whose centers are smooth complete intersections. Aluffi formula are center to this paper, and are reviewed in section \ref{sec:push}.

The pushforward allows to compute intersection numbers of a fibration in terms of the Chow ring  of its base. 
The first use of implicit pushforward methods in string theory is the computation of the Euler characteristic of a Weierstrass model by Sethi, Vafa and Witten for Calabi-Yau threefolds and fourfolds in terms of the Chow ring of its base \cite{Sethi:1996es}.  
Using efficient pushforward techniques, generating functions for the Euler characteristic of a smooth Weierstrass model was derived by Aluffi and Esole \cite{AE1}.
Intersection theory is instrumental in defining the induced D3-charges in presence of an orientifold  \cite{CDE}, which is to constrain to satisfy a ``tadpole condition''. 
See \cite{AE2,EFY,EKY} for additional examples of topological relations in the Chow ring motivated by string dualities and tapdole cancellations. 

The first application of Porteous formula in string theory is  \cite{Andreas:1999ng}. Aluffi's formula was used to compute the Euler characteristic of an SU(5) model  \cite{EY} in   \cite{Marsano} and in the study of anomaly cancellations in six dimensional gauge theories in  \cite{EK1,EKY,EKY2,ES}. 
Powerful methods to compute pushforwards in projective bundles using the functorial properties of the Segre class were obtained in \cite{FH2}, extended a point of view  used in \cite{AE1,AE2}. Recently in \cite{Euler}, we introduced new theorems that streamline the computation of pushforwards of blowups as simple algebraic manipulations involving  rational polynomials. These techniques allow a simple computations of pushforward of any analytic expression of the exceptional divisors of the blowup of a complete intersections of smooth divisors. 
All these techniques are illustrated in  \cite{Euler} by computing the Euler characteristics of crepant resolutions of Weierstrass models. 

\begin{thm}[Aluffi-Esole, {\cite{AE1}}]
For a smooth Weierstrass model, we have 
\begin{equation}\label{Eq.AE1}
 \varphi_*c(Y)= 12\ \frac{L}{1+6L} c(B).
 \end{equation}
\end{thm}

\begin{thm}[Esole--Jefferson--Kang,{\cite{Euler}}]
The \textnormal{SU($2$)}-model over a base of arbitrary dimension, gives 
 \begin{equation}\label{Eq.Euler}
 \varphi_*c(Y)= \frac{2 L+3 L S-S^2}{(1+S)(1+6L-2S)} c(B),
 \end{equation}
 where $S$ is the divisor supporting the fiber of type \textnormal{ I$_2$} or \textnormal{III}. 
 \end{thm}

\begin{thm}[Esole--Fullwood--Yau, {\cite[Theorem 1.4]{EFY}}]
Let $\varphi: Y\to B$ be an elliptic fibration defined by the 
complete intersection of two divisors of class $\mathscr{O}(2)\otimes\pi^* \mathscr{L}^{\otimes 2}$ in the projective bundle 
$\pi:\mathbb{P}\Big(\mathscr{O}_B\oplus \mathscr{L}\oplus \mathscr{L}\oplus\mathscr{L}\Big)\to B$.  
Then 
\begin{equation}\label{Eq.EFY}
\varphi_* c(Y)= \frac{4L(3+5L)}{(1+2L)^2} c(B),
\end{equation}
where  $L=c_1(\mathscr{L})$. 
\end{thm}

\begin{thm}[Esole--Kang--Yau, {\cite{EKY}}]
Consider a projective variety $B$ endowed with two line bundles $\mathscr{L}$ and $\mathscr{D}$. 
Let $\varphi:Y\to B$ be a smooth elliptic fibration defined as the  zero scheme  of a section of the  line bundle $\mathscr{O}(3)\otimes\pi^* \mathscr{L}^{\otimes 2}\otimes \pi^*\mathscr{D}$ in the projective bundle 
$\pi:\mathbb{P}\Big(\mathscr{O}_B\oplus \mathscr{L}\oplus \mathscr{D}\Big)\to B$. Then 
\begin{equation}\label{Eq.EKY}
\varphi_* c(Y)= 6 \frac{(2 L + 2 L^2 - L D + D^2)}{(1 + 2 L - 2 D) (1 + 2 L + D)} c(B),
\end{equation}
where  $D=c_1(\mathscr{D})$ and $L=c_1(\mathscr{L})$. 
\end{thm}

\subsection{Definitions and notations}
Throughout this paper, we work over the field of complex numbers. 
A variety is a reduced and irreducible algebraic scheme. 
 We denote the vanishing locus of the sections $f_1, \ldots, f_n$ by $V(f_1, \ldots, f_n)$. 
  The tangent bundle of a variety $X$ is denoted by $TX$ and the  normal bundle of a subvariety $Z$ of a  variety $X$ is denoted by  $N_Z X$. 
 Let $\mathscr{V}\rightarrow B$ be a vector bundle over a variety $B$. We denote the by $\mathbb{P}(\mathscr{V})$ the projective bundle of lines in  $\mathscr{V}$.
We use Weierstrass models defined with respect to the projective bundle $\pi : X_0 = \mathbb P[\mathscr O_B \oplus \mathscr L^{\otimes 2} \oplus \mathscr L^{\otimes 3} ] \rightarrow B$  where $\mathscr{L}$ is a line bundle of $B$. 
We denote the pullback of $\mathscr{L}$ with respect to $\pi$ by $\pi^* \mathscr{L}$. 
We denote by $\mathscr{O}_{X_0} (1)$  the  canonical line bundle on $X_0$, i.e., the  dual of the tautological line bundle of $X_0$  (see \cite[Appendix B.5]{Fulton.Intersection}).
The first Chern class of  $\mathscr{O}_{X_0} (1)$ is denoted $H$ and the first Chern class of $\mathscr{L}$ is denoted $L$. 
 The Weierstrass model $\varphi : Y_0 \rightarrow B$ is defined as  the zero-scheme of a section of $\mathscr{O}_{X_0} (3) \otimes \pi^* \mathscr L^{\otimes 6}$. 
 The Chow group $A_*(X)$ of a nonsingular variety $X$ is the group of divisors modulo rational equivalence \cite[Chap. 1,\S 1.3]{Fulton.Intersection}.
We use $[V]$ to refer to the class of a subvariety $V$ in $A_*(X)$.
 Given a class $\alpha \in A_*(X)$, the degree of $\alpha$ is denoted $\int_X \alpha$ (or simply $\int \alpha$ if $X$ is clear from the context.) Only the zero component of $\alpha$ is relevant in computing $\int_X \alpha$---see \cite[Definition 1.4, p. 13]{Fulton.Intersection}. We use $c(X)=c(TX)\cap [X]$ to refer to the total homological Chern class of a nonsingular variety $X$, and likewise we use $c_i(TX)$ to denote the $i$th Chern class of the tangent bundle $TX$.  
  Given   two varieties $X, Y$ and a proper morphism $f: X \rightarrow Y$, the  proper pushforward associated to 
$f$ is denoted $f_*$.  If $g: X\rightarrow Y$ is a flat morphism, the pullback of $g$ is denoted $g^*$ and by  definition $g^*[V]=[g^{-1} (V)]$, see  \cite[Chap 1, \S 1.7]{Fulton.Intersection}.
\begin{defn}[Pushforward,  {\cite[Chap. 1, p. 11]{Fulton.Intersection}}] 
 Let $f: X\longrightarrow Y$ be a proper morphism. 
Let $V$ be a subvariety of $X$, the image $W=f(V)$ a subvariety of $Y$, and the function field $R(V)$ an extension of the function field $R(W)$. 
The pushforward $f_* : A_*(X)\to A_*(Y)$ is defined as follows
$$
f_* [V]= \begin{cases}
0  &  \text{if} \quad \dim V\neq \dim W,\\
[R(V):R(W)] \   [V_2]  & \text{if} \quad \dim V= \dim W,
\end{cases}
$$
where $[R(V):R(W)]$ is the degree of the field extension $R(V)/R(W)$. 
\end{defn}
\begin{defn}[Degree, {\cite[Chap. 1, p. 13]{Fulton.Intersection}}]
The degree of a class $\alpha$ of $A_*(X)$ is denoted by  $\int_X \alpha$ (or simply $\int \alpha$ if there is no ambiguity in the choice of $X$), and is defined to be the degree of its component in $A_0(X)$.
\end{defn}
The total homological Chern class $c(X)$ of any nonsingular variety $X$ of dimension $d$ is defined as
\begin{equation*}
c(X)=c(TX)\cap [X],
\end{equation*}
where $TX$ is the tangent bundle of $X$ and $[X]$ is the class of $X$ in the Chow ring. The degree of $c(X)$ is the topological Euler characteristic of $X$: 
\begin{equation*}
\chi(X)=\int_X c(X).
\end{equation*}
The following Lemma gives an important functorial property of the degree.
\begin{lem}[{\cite[Chap. 1, p. 13]{Fulton.Intersection}}]  \label{lem:Push} Let $f:X\longrightarrow Y$  be a proper map between varieties. 
 For any class $\alpha$ in the Chow ring $A_*(X)$ of $X$:
$$\int_X \alpha=\int_Y f_* \alpha.$$
\end{lem}
 Lemma \ref{lem:Push} means that an intersection number in $X$ can be computed in $Y$ through a pushforward of a proper map $f:X\longrightarrow Y$. 
 This simple fact  has far-reaching consequences as it allows us to express the topological invariants of an elliptic fibration in terms of those of the base.

Let $X$ be a projective variety with at worst canonical Gorenstein singularities. 
We denote the canonical class by $K_X$. 
\begin{defn}
A birational projective  morphism $\rho:Y\longrightarrow X$ is called a \emph{crepant desingularization} of $X$ if $Y$ is smooth and 
$K_Y=\rho^* K_X$. 
\end{defn}

\begin{defn}
A resolution of singularities of a variety $Y$ is a proper surjective birational morphism $\varphi:\widetilde{Y}\longrightarrow Y$  such that  
$\widetilde{Y}$ is nonsingular
and  $\varphi$ is an isomorphism away  from the singular  locus of $Y$. In other words, $\widetilde{Y}$ is nonsingular and  if $U$ is the singular locus of $Y$, $\varphi$ maps $\varphi^{-1}(Y\setminus U)$  isomorphically  onto $Y\setminus U$.  
 A \emph{crepant resolution of singularities}  is a resolution of singularities such that  $K_Y=f^* K_X$. 
\end{defn}

\subsection{Pushforward formulas}
When pushing forward blowups of  a projective bundle $\pi: X_0=\mathbb{P}[\mathscr{O}_B\oplus\mathscr{L}^{\otimes 2} \oplus \mathscr{L}^{\otimes 3}]\longrightarrow B$, the key ingredients are the following three theorems. 
The first one is a theorem of Aluffi which gives the Chern class after a blowup along a local complete intersection. 
The second theorem is a pushforward theorem that provides a user-friendly method to compute invariant of the blowup space in terms of the original space. 
The last theorem is a direct consequence of functorial properties of the Segre class and gives a simple method to pushforward analytic expressions in the Chow ring of the projective bundle $X_0$ to  the Chow ring of its base. 
We follow mostly \cite{Euler}.

\begin{thm}[Aluffi, {
{\cite[Lemma 1.3]{Aluffi_CBU}}}]
\label{Thm:AluffiCBU}
Let $Z\subset X$ be the  complete intersection  of $d$ nonsingular hypersurfaces $Z_1$, \ldots, $Z_d$ meeting transversally in $X$.  Let  $f: \widetilde{X}\longrightarrow X$ be the blowup of $X$ centered at $Z$. We denote the exceptional divisor of $f$  by $E$. The total Chern class of $\widetilde{X}$ is then:
\begin{equation}
c( T{\widetilde{X}})=(1+E) \left(\prod_{i=1}^d  \frac{1+f^* Z_i-E}{1+ f^* Z_i}\right)  f^* c(TX).
\end{equation}
\end{thm}

\begin{lem}[See {\cite{Aluffi_CBU,Euler}}] \label{Thm:PushE}
\label{lem:symhom}
Let  $f: \widetilde{X}\longrightarrow X$ be the blowup of $X$ centered at $Z$. We denote the exceptional divisor of $f$  by $E$. Then 
\begin{equation}
 f_* E^n=
 (-1)^{d+1} h_{n-d} (Z_1, \ldots, Z_d) Z_1\cdots Z_d,\nonumber
\end{equation}
where $h_i(x_1, \ldots, x_k)$ is the complete homogeneous symmetric polynomial of degree $i$ in $(x_1, \ldots, x_k)$ with the convention that $h_i$ is identically zero for $i<0$ and $h_0=1$.  
\end{lem}

\begin{thm}[Esole--Jefferson--Kang,  see  {\cite{Euler}}] \label{Thm:Push}
    Let the nonsingular variety $Z\subset X$ be a complete intersection of $d$ nonsingular hypersurfaces $Z_1$, \ldots, $Z_d$ meeting transversally in $X$. Let $E$ be the class of the exceptional divisor of the blowup $f:\widetilde{X}\longrightarrow X$ centered 
at $Z$.
 Let $\widetilde{Q}(t)=\sum_a f^* Q_a t^a$ be a formal power series with $Q_a\in A_*(X)$.
 We define the associated formal power series  ${Q}(t)=\sum_a Q_a t^a$, whose coefficients pullback to the coefficients of $\widetilde{Q}(t)$. 
 Then the pushforward $f_*\widetilde{Q}(E)$ is
 $$
  f_*  \widetilde{Q}(E) =  \sum_{\ell=1}^d {Q}(Z_\ell) M_\ell, \quad \text{where} \quad  M_\ell=\prod_{\substack{m=1\\
 m\neq \ell}}^d  \frac{Z_m}{ Z_m-Z_\ell }.
 $$ 
\end{thm}

\begin{thm}[{See  \cite{Euler} }]\label{Thm:PushH}
Let $\mathscr{L}$ be a line bundle over a variety $B$ and $\pi: X_0=\mathbb{P}[\mathscr{O}_B\oplus\mathscr{L}^{\otimes 2} \oplus \mathscr{L}^{\otimes 3}]\longrightarrow B$ a projective bundle over $B$. 
 Let $\widetilde{Q}(t)=\sum_a \pi^* Q_a t^a$ be a formal power series in  $t$ such that $Q_a\in A_*(B)$. Define the auxiliary power series $Q(t)=\sum_a Q_a t^a$. 
Then 
$$
\pi_* \widetilde{Q}(H)=-2\left. \frac{{Q}(H)}{H^2}\right|_{H=-2L}+3\left. \frac{{Q}(H)}{H^2}\right|_{H=-3L}  +\frac{Q(0)}{6 L^2},
$$
 where  $L=c_1(\mathscr{L})$ and $H=c_1(\mathscr{O}_{X_0}(1))$ is the first Chern class of the dual of the tautological line bundle of  $ \pi:X_0=\mathbb{P}(\mathscr{O}_B \oplus\mathscr{L}^{\otimes 2} \oplus\mathscr{L}^{\otimes 3})\rightarrow B$.
\end{thm}

Since all the blowups used in this paper have  centers that are complete intersections of two or three smooth divisors, the following two Lemmas are all that is needed to compute pushforwards under such blowups when the base is a threefold. 
They are direct consequences of Lemma \ref{Thm:PushE}.
\begin{lem}\label{lem:Push2}
For a blowup  $f:\widetilde{X}\longrightarrow X$ with center a transverse intersection of two divisors of class $Z_1$ and $Z_2$, we have 
\begin{align}\nonumber
& f_* E=0, \qquad f_* E^2=-Z_1 Z_2, \qquad \   \   f_*E^3= -(Z_1 +Z_2)Z_1 Z_2, \\
 & f_* E^4=-(Z_1^2+Z_2^2+ Z_1 Z_2) Z_1 Z_2, 
\quad f_* E^5= -(Z_1+Z_2)(Z_1^2+Z_2^2)Z_1 Z_2 .
\end{align}
\end{lem}

\begin{lem}\label{lem:Push3}
For a blowup  $f:\widetilde{X}\longrightarrow X$ with center a transverse intersection of three divisors of class $Z_1$, $Z_2$, and $Z_3$, we have 
\begin{align}\nonumber
& f_* E=0, \qquad f_* E^2=0, \qquad  f_*E^3= Z_1 Z_2 Z_3, \\
&  f_* E^4=(Z_1+Z_2) Z_1 Z_2 Z_3, \quad   f_* E^5=  (Z_1^2+Z_2^2+Z_3^2+Z_1 Z_2+Z_1 Z_3+Z_2 Z_3)Z_1 Z_2 Z_3.
\end{align}
\end{lem}
\begin{lem}[\cite{Euler,Fullwood:SVW}] \label{lem:PushH}
Given the projective bundle $ \pi:X_0=\mathbb{P}(\mathscr{O}_B \oplus\mathscr{L}^{\otimes 2} \oplus\mathscr{L}^{\otimes 3})\rightarrow B$, denoting the first Chern class of $\mathscr{L}$ by $L$, we have:
$$
\begin{aligned}
\pi_* 1 &= 0, \quad \pi_*H= 0, \quad \pi_*H^2= 1 , \quad  \pi_* H^3 = -5 L, \quad\pi_*  H^4= 19 L^2, \quad \pi_* H^5= -65 L^3,\\
\pi_* H^{k} &=\left[(-2)^{k-1} -(-3)^{k-1} \right] L^{k-2} \quad n\geq 1 \nonumber.
\end{aligned}
$$
\end{lem}
Using Lemma \ref{lem:PushH}, it is then direct to show that 
\begin{equation}
\pi_* \Big(H^k (3H+6 L)\Big)= -(-3)^k  L^{k-1}, \quad k\geq 1
\end{equation}
\subsection{Example: the second Chern class under a crepant birational map}
We now look at the behavior of the second Chern class of a divisor $V$ under a crepant blowup whose center is a complete intersection of $n$ smooth  divisors intersecting transversally in a smooth  ambient space $X$.

\begin{thm}\label{Thm.C2}
Consider a smooth variety $X$ and a blowup $f:\widetilde{X}\to X$ with center  the complete intersection of $n\geq 2$ smooth divisors $Z_i$ intersecting transversally. We denote  the exceptional divisor by $E$. Consider a divisor  $V$ in $X$. If we ask for the restriction of the blowup to $V$ to be crepant, then the class of $V$ is such that its  proper transform  is  $\widetilde{V}=f^* V-(n-1)E$. It follows that 
$$
f_* \Big(c_2 (T\widetilde{V})\cdot  \widetilde{V} \Big) =c_2(TV) \cdot V.
$$
\end{thm}
\begin{proof}
Using Aluffi's formula (Theorem \ref{Thm:AluffiCBU}) and the adjunction formula, we have 
$$
c(T\widetilde{V})=(1+E)\Big({\prod_{i=1}^n \frac{(1 +f^* Z_i -E)}{(1+f^*Z_i)}}\Big) \frac{f^* c(TX)}{ (1+f^* V-(n-1) E)}.
$$ 
In the case $n-2$, by a direct expansion: 
$$
c_1(T\widetilde{V})=f^*c_1(TX)-f^*V=f^* c_1 (TV), \quad c_2(T\widetilde{V})= f^* c_2(Y)-V(Y-Z_1-Z_2) -V^2,
$$
where   $c_2(TV)=c_2(TX)-c_1 V+ V^2$. It follows that 
$$f_*c_2 (T\widetilde{V})\cap \widetilde{V}=f_* \Big[f^* c_2(TV)\cap [V] +f^*Vf^*( Z_1 + Z_2) E - f^*(Z1+Z2)E^2 + E^3]$$
By applying Lemma \ref{lem:Push2}, we find 
$$f_*c_2 (T\widetilde{V})\cap (\widetilde{V})= c_2(TV)\cap [V].$$

In the case $n=3$, there is no contribution of order $E^3$ in $f_*c_2 (T\widetilde{V})\cap (\widetilde{V})$ and by Lemma \ref{lem:Push3} we have $ f_* E=f_* E^2=0$. 

$$
\begin{aligned}
f_* \Big(c_2 (T\widetilde{V})\cdot  \widetilde{V} \Big) &=f^*[c_2(V) \cdot V]
\\ & \quad \  + (2 c_1 V-2 c_2-4 V^2+V Z_1+V Z_2+V Z_3) E+(4 V-2 Z_1-2 Z_2-2 Z_3)E^2.
\end{aligned}
$$
In the case $n\geq 4$, the theorem is trivial because  
$c_2 (T\widetilde{V})\cdot  \widetilde{V}$ is at most cubic in $E$ and  $f^* E^i=0$ for  $i\leq 3$ by Lemma \ref{Thm:PushE}. 
\end{proof}

\subsection{Example: Todd class of a flat fibration of genus-$g$ curves. }\label{sec:Toddg}
\begin{thm}[Esole-Fullwood-Yau, {\cite[Theorem A.1]{EFY}}]
Let $\varphi:X\to B$ be a proper and flat morphism between smooth projective varieties such that the generic fiber of $\varphi$ is a smooth curve of genus $g$. Then 
\begin{equation}
\varphi_* Td (X)=  \Big(1-ch(\varphi_* \omega_{X/B}^\vee) \Big) Td (B),
\end{equation}
where $\omega_{X/B}$ is the relative dualizing sheaf of the fibration. 
\end{thm}
When the variety $Y$ is smooth, $\omega_{X/B}=\omega_Y\otimes (\varphi^*\omega_B)^\vee$. 
In particular, in the case of an elliptic fibration, we get 
\begin{equation}
\varphi_* Td (Y)=  (1-e^{-L}) Td (B),
\end{equation}
where $L=c_1(\mathscr{L})$ and $\mathscr{L}$ is the fundamental line bundle of the Weierstrass model. 

The previous theorem  shows that the holomorphic Euler characteristic depends only on the base and the line bundle $\mathscr{L}$. 
In particular, for a $G$-model, it does not depend on the Kodaira type. %

the holomorphic arithmetic genus of $Y$ is the same as $\chi(W,\mathscr{O}_W)$ where $W$ is the Weierstrass divisor defined as the zero locus of a smooth section of $\mathscr{L}$. See 
\cite[Appendix A]{EFY}.
\begin{equation}
\chi_0(Y)=\int_Y \varphi_* Td (Y)= \chi(W,\mathscr{O}_W).
\end{equation}

\subsection{Example: Smooth Weierstrass model}
The Euler characteristic of a smooth Weierstrass model $Y_0\to B$ with fundamental line bundle $\mathscr{L}$ is 
\begin{equation}
\chi(Y_0)=12L\int_B \frac{1}{1+6L} c(TB),\quad \text{where}\   L=c_1(\mathscr{L}). 
\end{equation}

For an elliptic fourfold given by a smooth Weierstrass model:
\begin{align}
\chi_0(Y_0) &= \frac{1}{12}\int_B L (c_1^2+c_2-3 c_1 L+2 L^2),\quad
\chi_1(Y_0) =
-\frac{1}{3}\int_B L (2 c_1^2+5 c_2-54 c_1 L+232 L^2),\\
\chi_2(Y_0) &=  
-\frac{1}{2} \int_B L (3 c_1^2-17 c_2+71 c_1 L-554 L^2).
\end{align}
Using the adjunction formula, we have
\begin{equation}
c(TY_0)= \frac{(1+H)(1+H+2L)(1+H+3L)}{(1+3H+6L)} (1+c_1(TB)+c_2(TB)+c_3(TB))
\end{equation}
 We abuse notation and do not write the pullback. By expanding, we get: 
\begin{align}\label{Eq.ExpC}
\begin{aligned}
c_1(TY_0)&= (c_1-L),\quad c_2(TY)=c_2-c_1 L+13 H L+12 L^2+3 H^2, \\
c_3(TY_0)&=-72 L^3+12 c_1 L^2-c_2 L+c_3+H (13 c_1 L-108 L^2)+H^2 (3 c_1-52 L)-8 H^3\\
c_4(TY_0)&=(-72 c_1 L^3 + 12 c_2 L^2 - c_3 L + 432 L^4)+H (-108 c_1 L^2 + 13 c_2 L + 864 L^3)\\
& \quad + (-52 c_1 L+3 c_2+636 L^2)H^2+(204 L - 8 c_1)H^3 +24 H^4.
\end{aligned}
\end{align}
\begin{thm}[Chern and Pontryagin numbers of a smooth Weierstrass fourfold]\label{Thm.CP.4}
Let $B$ be a projective threefold and $\varphi:Y_0\to B$ be a smooth Weierstrass model with fundamental line bundle $\mathscr{L}$. Denoting the first Chern class of $\mathscr{L}$ by $L$  
and writing the $i$th Chern class of the base $B$ simply as  $c_i$, 
the Chern and Pontryagin numbers of $Y_0$ are
\begin{align}
 \int_{Y_0} c_1^4(TY_0)&=0, \\
 \int_{Y_0} c_1^2(TY_0) c_2(TY_0) &=12 \int_BL (c_1-L)^2, \\
 \int_{Y_0} c_2^2(TY_0) &=24  \int_BL (6 L^2-c_1 L+c_2),\\
 \int_{Y_0}c_1(TY_0) c_3(TY) &= 12  \int_BL (c_1-6 L) (c_1-L),\\
\int_{Y_0} c_4(TY_0) &=12 \int_B L (36 L^2-6 L c_1+c_2),
\end{align}
\begin{align}
 \int_{Y_0}p_2 (TY_0)&=-24  \int_BL (-c_1^2+2 c_2+36 L^2),\\
 \int_{Y_0} p_1^2(TY_0) &=-48 \int_B L (c_1^2-2 c_2-11 L^2).
\end{align}
\end{thm}
\begin{proof}
For each entries, we use \eqref{Eq.ExpC} and compute 
$$\int_{Y_0} A= \int_{X_0} A \cdot (3 H + 6L)=3\int_B \pi_*\Big[A \cdot  (H+2L)\Big].$$ 
The pushforward is then evaluated using Theorem \ref{Thm:PushH} or Lemma \ref{lem:PushH}. 

For instance, $\int_{Y_0} c_4(TY_0)=\int_{X_0} c_4(TY_0)(3H+6L)$. Since terms independent of $H$ or linear in $H$ will not contribute (see Lemma \ref{lem:PushH}), we have 
\begin{align}
\int_{Y_0} c_4(TY_0)&=\int_{X_0} c_4(TY_0)(3H+6L) \nonumber \\ 
&=\int_B \pi_* \Big(H^3 (-204 c_1 L + 9 c_2 + 3132 L^2) + 
 H^2 (-636 c_1 L^2 + 57 c_2 L + 6408 L^3) + \nonumber \\
 &\quad\quad  H^4 (756 L - 24 c_1) + 72 H^5\Big)\nonumber\\
 &=12 \int_B L (36 L^2-6 L c_1+c_2)\nonumber
 \end{align}
 Using  Lemma \ref{lem:PushH}, we implement the pushforward with the substitution $H^2\to 1$, $H^3\to -5L$, $H^4\to 19 L^2$, $H^5\to -65 L^3$, 
 which gives the final answer. 
\end{proof}
 As a direct application of Theorem \ref{Thm.CP.4}, we compute the following invariants for a smooth Weierstrass model. 
 \begin{thm}[$L$-genus, $\hat{A}$-genus, and $X_8$ invariant of a smooth Weierstrass model]
 Let $B$ be a projective threefold and $\varphi:Y_0\to B$ be a smooth Weierstrass model with fundamental line bundle $\mathscr{L}$. Denoting the first Chern class of $\mathscr{L}$ by $L$ and the $i$th Chern class of the base $B$ simply by  $c_i$, then, we have
\begin{align}
45\sigma &=\int_{Y_0}\Big(7p_2(TY_0)-p_1^2(TY_0)\Big)\ \   =120\int_B  L(-c_1^2+2 c_2+46 L^2),\\
5760\int_{Y_0}\hat{A}_2&=\int_{Y_0}\Big(7p_1^2(TY_0)-4 p_2(TY_0)\Big)=240\int_B  L(-c_1^2+2 c_2+L^2),\\
192X_8&=\int_{Y_0}\Big(p_1^2(TY_0)-4p_2(TY_0)\Big)\ \  =48\int_B  L (c_1^2-2 c_2-61 L^2).
\end{align}
\end{thm}
\begin{proof}
True by linearity from the quantities computed in Theorem \ref{Thm.CP.4}. 
\end{proof}

\subsection{Example: the SU($2$)-model}
In this section, we discuss in detail the computation of Chern numbers  of the SU(2)-model. The result is independent of a choice of a possible Kodaira fiber realizing the Dynkin diagram of type A$_1$ because  type I$_2$ and III are resolved by the same blowup. The Weierstrass equations defining SU(2)-models are given in section \ref{sec:Tate}. The Weierstrass equation is defined in the ambient space as the projective bundle $X_0$, where the projection map is $\pi: X_0=\mathbb{P}_B[\mathscr{O}_B\oplus \mathscr{L}^{\otimes 2}\oplus \mathscr{L}^{\otimes 3}]\to B$ and the defining equation is a section of $\mathscr{O}(3)\otimes\mathscr{L}^{\otimes 6}$.
Hence, we find 
\begin{align}
c(TX_0)&=(1+H)(1+H+3\pi^* L)(1+H+2\pi^* L) \pi^*c(TB),\nonumber\\
c(TY_0)&=\frac{c(X_0)}{1+3H+6\pi^* L},\nonumber
\end{align}
 where  $L=c_1(\mathscr{L})$ and $H=c_1(\mathscr{O}_{X_0}(1))$ is the first Chern class of the dual of the tautological line bundle of  $ \pi:X_0=\mathbb{P}(\mathscr{O}_B \oplus\mathscr{L}^{\otimes 2} \oplus\mathscr{L}^{\otimes 3})\rightarrow B$. 
 We denote by $S=V(s)$ the Cartier divisor supporting the fiber I$_2$ or III with dual graph of Dynkin type $\widetilde{A}_1$. 

The singular elliptic fibration of an SU(2)-model is resolved by a unique blowup with the center $(x,y,s)$, which we denote as $f: X_1 \longrightarrow X_0$ with the exceptional divisor $E_1$.
The center is a complete intersection of hypersurfaces $V(x)$, $V(y)$, and $V(s)$, whose classes are respectively 
\begin{equation}\nonumber
 Z_1= 2\pi^*L+H, \quad Z_2= 3\pi^* L+H, \quad Z_3=\pi^* S.
\end{equation}
The proper transform  of the elliptic fibration $Y_0$ is denoted as $Y$, and is obtained from the total transform of $Y$ by removing $2E_1$. It follows that the class of $Y$ in $X_1$ is 
\begin{align}\nonumber
[Y]=[f^*(3H+6 \pi^* L)-2 E_1]\cap[X_1].
\end{align}
Using Theorem \ref{Thm:AluffiCBU}, we have the following Chern class for $X_1$
\begin{align} \nonumber
c(TX_1)&=(1+E_1)\frac{(1+f^* Z_1-E_1)(1+f^* Z_2 -E_1)(1+f^* Z_3-E_1)}{(1+f^* Z_1)(1+f^* Z_2)(1+f^* Z_3)} f^* c(TX_0).
\end{align}
The adjunction formula gives
\begin{align} \nonumber
c(TY)&=\frac{(1+E_1)(1+f^* Z_1-E_1)(1+f^* Z_2 -E_1)(1+f^* Z_3-E_1)}{(1+3f^*H+6f^*\pi^* L-2E_1)(1+f^* Z_1)(1+f^* Z_2)(1+f^* Z_3)} f^* c(TX_0).
\end{align}
Concretely, we replace $c(TB)$ by the Chern polynomial $c_t(TB)=1+c_1t + c_2t^2 + c_3 t^3$, $L$ by $ Lt$, and $S$ by $ St$.  Then, $c_i(TY)$ is given by the coefficient of $t^i$ in the Taylor expansion centered at $t=0$.
 From the expansion, we get the following expression (to ease the notation, we do not write the pullbacks. Thus, by $c_i$, $L$, $S$, and $H$ we mean $f^* \pi^* c_i$, $f^*\pi^* L$, $f^*\pi^* S$, and $f^*H$ respectively).
\begin{align}
c_1(TY) &= (c_1-L),\\
c_2(TY) &= (c_2-c_1 L+12 L^2)+3 H^2+13 H  L+E_1(-4  H-7   L+S),\\
c_3(TY) &= 3 c_1 H^2+13 c_1 H L+12 c_1 L^2-c_2 L+c_3-8 H^3-52 H^2 L-108 H L^2-72 L^3\nonumber\\
&\  \   +E_1 (-4 c_1 H-7 c_1 L+c_1 S+16 H^2+66 H L+66 L^2-L S-S^2)\\
&\  \   +E_1^2 (-10 H-19 L+S)+2 E_1^3,\nonumber\\
c_4(TY) & = (-8 c_1 H^3-52  H^2c_1 L-108 c_1 H L^2-72 c_1 L^3+3 c_2 H^2+13 c_2 H L\nonumber \\
&\  \  +12 c_2 L^2-c_3 L+24 H^4+204 H^3 L+636 H^2 L^2+864 H L^3+432 L^4)\nonumber\\
&\  \   + E_1(16 c_1 H^2+66 c_1 H L+66 c_1 L^2-c_1 L S-c_1 S^2-4 c_2 H-7 c_2 L+c_2 S  \\
&\   \  -64 H^3-398 H^2 L+3 H^2 S-810 H L^2+13 H L S-540 L^3+12 L^2 S+L S^2+S^3)\nonumber\\
&\  \ + E_1^2 (-10 c_1 H-19 c_1 L+c_1 S+59 H^2+239 H L-4 H S+240 L^2-8 L S-S^2) \nonumber\\
&\   \  + E_1^3( 2 c_1-22 H-45 L+S)+3 E_1^4.\nonumber
\end{align}
We can now compute $\int_Y c_2^2(TY)$ and apply the pushforward formula of Lemma \ref{lem:Push3} 
\begin{align}
\begin{split}
\int_{Y} c_2^2(TY)
=& \int_{X_1}  c_2^2(TY) [Y]=\int_{X_1}  c_2^2(TY) (f^*(3H+6 \pi^* L)-2 E_1)\\
 =&\int_{X_0} f_*\Big( c_2^2(TY) (f^*(3H+6 \pi^* L)-2 E_1)\Big) \\
  =&\int_{B}\pi_* f_*\Big( c_2^2(TY) (f^*(3H+6 \pi^* L)-2 E_1)\Big) \\
=&24\int_B L (c_2-c_1L+6 L^2)-2\int_BS(7L-S)^2.
\end{split}
\end{align}
 In the second line, we use $\int_Y A=\int_{X_1} A\cap [Y]$. In the third line, we use the functorial property $\int_{X_1}A=\int_{X_0} f_* A$ for a proper map $f:X_1\to X_0$. 
In the fourth line, we use again the functional property of the degree, but this time, for the proper map $\pi: X_0\to B$. 
In the last line, we use Lemma \ref{lem:Push3} to compute the pushforward of $f$ to $X_0$ and then  Lemma \ref{lem:PushH} to compute of $\pi$ to the base $B$.

Using the same logic, we now  compute $\int_Y c_3(TY) f^*\pi^*\alpha$. 
\begin{align}
\begin{split}
\int_Yc_3(TY) f^*\pi^*\alpha &=\int_{X_1}c_3(TY)  (f^*(3H+6 \pi^* L)-2 E_1) f^*\pi^*\alpha\\
&=\int_{X_0} f_* \Big(c_3(TY) (3H+6L-2E_1)\Big) \pi^* \alpha\\
&=\int_{B}\pi_*f_* \Big(c_3(TY) (3H+6L-2E_1)\Big) \alpha\\
&=\int_B 12L(c_1-6L)\alpha +6\int_B (5L-S) S \alpha.
\end{split}
\end{align}
 In the first line, we use $\int_Y A=\int_{X_1} A\cap [Y]$.
The second and third lines are projection formula for $f$ and $\pi$.
In the last line, we use Lemma \ref{lem:Push3} to compute the pushforward of $f$ to $X_0$ and then  Lemma \ref{lem:PushH} to compute of $\pi$ to the base $B$. 

\section{$G$-models}\label{sec:GModels}

\subsection{Tate's forms for $G$-models}\label{sec:Tate}
The variable $s$ is a section of the line bundle $\mathscr{O}_B(S)$. In other words, the zero locus locus $V(s)$ is the divisor $S$ supporting the singular fiber.
\begin{align}
\label{eqn:SO}
&\text{I}_2^{\text{s}} &&\text{SU(2)}&:&\quad 	&&y^2z+a_1 x y z + a_{3,1} s y z =x^3 + a_{2,1} s  x^2z + a_{4,1}s x z^2+a_{6,2} s^2 z^3,\\
&\text{I}_{2n}^{\text{ns}} &&\text{USp($2n$)}&:&\quad 	&&y^2z =x^3 + a_2 x^2z + a_{4,n} s^n x z^2+a_{6,2n} s^{2n} z^3,\\
&\text{I}_{2n+1}^{\text{ns}} &&\text{USp($2n$)}&:&\quad 	&&y^2z =x^3 + a_2 x^2z + a_{4,n+1} s^{n+1} x z^2+ a_{6,2n+1} s^{2n+1} z^3,\\
&\text{I}_{2n}^{\text{s}} &&\text{SU($2n$)}&:&\quad 	&&y^2z+ a_1 x y z =x^3 + a_{2,1} s x^2z + a_{4,n} s^n x z^2+a_{6,2n} s^{2n} z^3,\\
&\text{I}_{2n+1}^{\text{s}} &&\text{SU($2n+1$)}&:&\quad 	&&y^2z+a_1 x y z+a_{3, n} s^n y z^2   =x^3 + a_{2,1}s x^2z + a_{4,n+1} s^{n+1} x z^2+a_{6,2n+1} s^{2n+1} z^3,\\
&\text{I}_0^{*\text{ss}}\quad &&\text{Spin(7)}&:& \quad && y^2z= x^3 + a_{2,1} s x^2 z+ a_{4,2} s^2 x z^2+ a_{6,4} s^{4} z^3,\\
&\text{I}_0^{*\text{s}}\quad &&\text{Spin(8)}&:& \quad && y^2z= (x-x_1sz)(x-x_2 sz)(x-x_3 sz)+  s^2 r x^2z + s^3 q x z^2+ s^4 t z^3, \\
&\text{III}&&\text{SU(2)}&:&\quad &&y^2z  =x^3  + s a_{4,1} xz^2 + s^2 a_{6,2} z^3,\\
&\text{IV}^{\text{ns}} &&\text{SU(2)}&:&\quad  &&y^2z  =x^3 +s^2 a_{4,2} xz^2 + s^2 a_{6,2} z^3,\\
&\text{IV}^{\text{s}}\quad  &&\text{SU(3)}&:& \quad &&  y^2z + a_{3,1} s y z^2 =x^3  + s^2 a_{4,2} xz^2 + s^3a_{6,3} z^3,\\
	&\text{I}_0^{*\text{ns}}\quad  &&\text{G}_2&:& \quad &&y^2z=x^3+s^{2} a_{4,2}xz^2 +s^3 a_{6,3}z^3,\\
&\text{IV}^{*\text{ns}} \quad &&\text{F}_4&:&\quad 	&&y^2z =x^3 + s^{3} a_{4,3}xz^2 +s^4 a_{6,4}  z^3,\\
&\text{IV}^{*\text{s}} \quad\ &&\text{E}_6&:&\quad &&y^2z+a_{3,2} s^2 y z^2   = x^3 +s^{3}  a_{4,3} x z^2 +  s^5 a_{6,5}z^3,\\
&\text{III}^{*} \quad\   &&\text{E}_7&:&\quad  &&y^2z =x^3 +s^3 a_{4,3}  xz+s^{5} a_{6,5} z^3, \\
&\text{II}^{*} \quad\  &&\text{E}_8&:& \quad &&y^2z=x^3+s^4 a_{4,4} xz^2 +s^5 a_{6,5} z^3.
\end{align}

\begin{table}[hbt]
\begin{center}
\scalebox{0.88}{$			
\begin{array}{|c | c | c |} 
\hline
\text{Algebra} & \text{Dynkin diagram} & \text{Kodaira type}
\\\hline 
 &			\scalebox{.97}{$\begin{array}{c}\begin{tikzpicture}
				\node[draw,circle,thick,scale=1.25,label=above:{1}, fill=black] (1) at (0,0){};
			\end{tikzpicture}\end{array}
		$} 		
				& \text{I}_1, \quad \text{II}\\\hline

			\text{SU}(2)
 &			\scalebox{.97}{$\begin{array}{c}\begin{tikzpicture}
				\node[draw,circle,thick,scale=1.25,label=above:{1}, fill=black] (1) at (0,0){};
				\node[draw,circle,thick,scale=1.25,label=above:{1}] (2) at (1.3,0){};
				\draw[thick] (0.15,0.1) --++ (.95,0);
				\draw[thick] (0.15,-0.09) --++ (.95,0);
				
			\end{tikzpicture}\end{array}
		$} 		& \text{I}_2^\text{s}, \quad \text{I}_2^{\text{ns}}, \quad\text{I}_3^{\text{ns}},\quad \text{III},\quad \text{IV}^{\text{ns}}
				\\\hline
{ \text{SU}(\ell) \ (\ell\geq 3)}
& \scalebox{.97}{$\begin{array}{c} \begin{tikzpicture}
				\node[draw,circle,thick,fill=black,scale=1.25,label=above:{1}] (0) at (90:1.1){};	
				\node[draw,circle,thick,scale=1.25,label=above:{1}] (1) at (-2,0){};
				\node[draw,circle,thick,scale=1.25,label=below:{1}] (2) at (-1,0){};
				\node[draw,circle,thick,scale=1.25,label=below:{1}] (3) at (1,0){};	
				\node[draw,circle,thick,scale=1.25,label=above:{1}] (4) at (2,0){};	
				\draw[thick]   (2)--(1)--(0)--(4)--(3);
\draw[ultra thick, loosely dotted] (2) to (3) {};

			\end{tikzpicture}\end{array}$}            &{ \text{I}_{\ell}^\text{s}   }    \\\hline

{\text{Spin}(2(4+\ell )) \ (\ell\geq 0)}
						& 
						\scalebox{.97}{$\begin{array}{c} \begin{tikzpicture}
				\node[draw,circle,thick,scale=1.25,fill=black,label=above:{1}] (1) at (-.1,.7){};
				\node[draw,circle,thick,scale=1.25,label=above:{1}] (2) at (-.1,-.7){};	
				\node[draw,circle,thick,scale=1.25,label=above:{2}] (3) at (1,0){};
				\node[draw,circle,thick,scale=1.25,label=above:{2}] (4) at (2.1,0){};
				\node[draw,circle,thick,scale=1.25,label=above:{2}] (5) at (3.3,0){};	
				\node[draw,circle,thick,scale=1.25,label=above:{1}] (6) at (4.6,.7){};	
								\node[draw,circle,thick,scale=1.25,label=above:{1}] (7) at (4.6,-.7){};	
				\draw[thick] (1) to (3) to (4);
				\draw[thick] (2) to (3);
				\draw[ultra thick, loosely dotted] (4) to (5) {};
				\draw[thick] (5) to (6); 
				\draw[thick] (5) to (7);
			\end{tikzpicture}\end{array}$}&  \text{I}_{\ell}^{*\text{s}}    
												  \\\hline
												  
\text{E}_{6} 
&			\scalebox{.92}{$\begin{array}{c}\begin{tikzpicture}
				\node[draw,circle,thick,scale=1.25,fill=black,label=below:{1}] (0) at (0,0){};
				\node[draw,circle,thick,scale=1.25,label=below:{2}] (1) at (1,0){};
				\node[draw,circle,thick,scale=1.25,label=below:{3}] (2) at (1*2,0){};
				\node[draw,circle,thick,scale=1.25,label=below:{2}] (3) at (1*3,0){};
				\node[draw,circle,thick,scale=1.25,label=below:{1}] (4) at (1*4,0){};
								\node[draw,circle,thick,scale=1.25,label=left:{2}] (5) at (1*2,1*1){};
																\node[draw,circle,thick,scale=1.25,label=above:{1}] (6) at (1*2,1*2){};
				\draw[thick] (0)--(1)--(2)--(3)--(4);
				\draw[thick] (2)--(5)--(6);

			\end{tikzpicture}\end{array}$}  &  \text{IV}^{*\text{s}}   
		
			\\\hline
\text{E}_{7}	
								& \scalebox{.94}{$\begin{array}{c}
 \begin{tikzpicture}
				\node[draw,circle,thick,scale=1.25,fill=black, label=below:{1}] (0) at (0,0){ };
				\node[draw,circle,thick,scale=1.25,label=below:{2}] (1) at (1,0){ };
				\node[draw,circle,thick,scale=1.25,label=below:{3}] (2) at (1*2,0){ };
				\node[draw,circle,thick,scale=1.25,label=below:{4}] (3) at (1*3,0){ };
				\node[draw,circle,thick,scale=1.25,label=below:{3}] (4) at (1*4,0){ };
								\node[draw,circle,thick,scale=1.25,label=below:{2}] (5) at (5,0){ };
												\node[draw,circle,thick,scale=1.25,label=below:{1}] (6) at (6,0){ };
																\node[draw,circle,thick,scale=1.25,label=above:{2}] (7) at (3,1){ };
																								
				\draw[thick] (0)--(1)--(2)--(3)--(4)--(5)--(6);
				\draw[thick] (3)--(7);

			\end{tikzpicture}
			\end{array}$}		&  \text{III}^{*}    \\\hline	  
\text{E}_{8}							  
						&
						 \scalebox{.94}{
 $\begin{array}{c}
 \begin{tikzpicture}
				\node[draw,circle,thick,scale=1.25,fill=black,label=below:{1}] (0) at (0,0){};
				\node[draw,circle,thick,scale=1.25,label=below:{2}] (1) at (1*1,0){ };
				\node[draw,circle,thick,scale=1.25,label=below:{3}] (2) at (1*2,0){ };
				\node[draw,circle,thick,scale=1.25,label=below:{4}] (3) at (1*3,0){ };
				\node[draw,circle,thick,scale=1.25,label=below:{5}] (4) at (1*4,0){ };
								\node[draw,circle,thick,scale=1.25,label=below:{6}] (5) at (1*5,0){ };
												\node[draw,circle,thick,scale=1.25,label=below:{4}] (6) at (1*6,0){ };
																\node[draw,circle,thick,scale=1.25,label=above:{3}] (8) at (1*5,1*1){ };
																\node[draw,circle,thick,scale=1.25,label=below:{2}] (7) at (1*7,0){ };
								
				\draw[thick] (0)--(1)--(2)--(3)--(4)--(5)--(6)--(7);
				\draw[thick] (5)--(8);

			\end{tikzpicture}
			\end{array}$}  &  {\text{II}^{*} }  \\\hline
												  
												   \begin{array}{c}
												  \\
{\text{Spin}(2(3+\ell)+1) \ (\ell\geq 0)}   \\
						 \\
						 \end{array}
						 &\scalebox{.97}{$\begin{array}{c} \begin{tikzpicture}
				\node[draw,circle,thick,scale=1.25,fill=black,label=above:{1}] (1) at (-.1,.7){};
				\node[draw,circle,thick,scale=1.25,label=above:{1}] (2) at (-.1,-.7){};	
				\node[draw,circle,thick,scale=1.25,label=above:{2}] (3) at (1,0){};
				\node[draw,circle,thick,scale=1.25,label=above:{2}] (4) at (2.1,0){};
				\node[draw,circle,thick,scale=1.25,label=above:{2}] (5) at (3.3,0){};	
				\node[draw,circle,thick,scale=1.25,label=above:{1}] (6) at (4.6,0){};	
				\draw[thick] (1) to (3) to (4);
				\draw[thick] (2) to (3);
				\draw[ultra thick, loosely dotted] (4) to (5) {};
				\draw[thick] (3.5,-0.05) --++ (.9,0){};
				\draw[thick] (3.5,+0.05) --++ (.9,0){};
				\draw[thick]
					(3.9,0) --++ (60:.25)
					(3.9,0) --++ (-60:.25);
			\end{tikzpicture}\end{array}$}		& {\begin{cases} \text{I}_{0}^{*\text{ss}}  \quad \text{for} \quad \ell= 0 \\  	 \text{I}_{\ell}^{*\text{ns}}\quad  \text{for}\quad \ell\geq 1\end{cases} }
			\\\hline
			
{\text{USp}(2(2+\ell)) \ (\ell\geq 0)}
						   &
			\scalebox{.97}{$\begin{array}{c} \begin{tikzpicture}
				\node[draw,circle,thick,scale=1.25,fill=black,label=above:{1}] (1) at (-.2,0){};
				\node[draw,circle,thick,scale=1.25,label=above:{1}] (3) at (.8,0){};
				\node[draw,circle,thick,scale=1.25,label=above:{1}] (4) at (1.8,0){};
				\node[draw,circle,thick,scale=1.25,label=above:{1}] (5) at (2.8,0){};	
				\node[draw,circle,thick,scale=1.25,label=above:{1}] (6) at (3.8,0){};	
				\node[draw,circle,thick,scale=1.25,label=above:{1}] (7) at (4.8,0){};	
				\draw[thick]   (3) to (4);
				\draw[thick] (5) to (6);
				\draw[ultra thick, loosely dotted] (4) to (5) {};
				\draw[thick] (4.,-0.05) --++ (.6,0){};
				\draw[thick] (4.,+0.05) --++ (.6,0){};
								\draw[thick] (-.2,-0.05) --++ (.8,0){};
				\draw[thick] (-.2,+0.05) --++ (.8,0){};

				\draw[thick]
					(4.4,0) --++ (-120:.25)
					(4.4,0) --++ (120:.25);
					\draw[thick]
					(0.2,0) --++ (-60:.25)
					(0.2,0) --++ (60:.25);
			\end{tikzpicture}\end{array}$}
& {\text{I}_{4+2\ell}^{\text{ns}},\quad    \text{I}_{5+2\ell}^{\text{ns}}}

			\\\hline
\text{F}_4
&			\scalebox{.97}{$\begin{array}{c}\begin{tikzpicture}
				\node[draw,circle,thick,scale=1.25,fill=black,label=above:{1}] (1) at (0,0){};
				\node[draw,circle,thick,scale=1.25,label=above:{2}] (2) at (1,0){};
				\node[draw,circle,thick,scale=1.25,label=above:{3}] (3) at (2,0){};
				\node[draw,circle,thick,scale=1.25,label=above:{2}] (4) at (3,0){};
				\node[draw,circle,thick,scale=1.25,label=above:{1}] (5) at (4,0){};
				\draw[thick] (1) to (2) to (3);
				\draw[thick]  (4) to (5);
				\draw[thick] (2.2,0.05) --++ (.6,0);
				\draw[thick] (2.2,-0.05) --++ (.6,0);
				\draw[thick]
					(2.4,0) --++ (60:.25)
					(2.4,0) --++ (-60:.25);
			\end{tikzpicture}\end{array}$} &  \text{IV}^{*\text{ns}}
				
				\\\hline

\text{G}_2
 &			\scalebox{.95}{$\begin{array}{c}\begin{tikzpicture}
				\node[draw,circle,thick,scale=1.25,label=above:{1}, fill=black] (1) at (0,0){};
				\node[draw,circle,thick,scale=1.25,label=above:{2}] (2) at (1.3,0){};
				\node[draw,circle,thick,scale=1.25,label=above:{1}] (3) at (2.6,0){};
				\draw[thick] (1) to (2);
				\draw[thick] (1.5,0.09) --++ (.9,0);
				\draw[thick] (1.5,-0.09) --++ (.9,0);
				\draw[thick] (1.5,0) --++ (.9,0);
				\draw[thick]
					(1.9,0) --++ (60:.25)
					(1.9,0) --++ (-60:.25);
			\end{tikzpicture}\end{array}
		$} 		&  \text{I}_0^{*\text{ns}}

				\\\hline\end{array}$}
	\end{center}
	\caption{
{ Affine Dynkin diagrams appearing as dual graphs of decorated Kodaira fibers. 
 } \label{Table:Affine}
}
\end{table}
\clearpage
\subsection{Crepant resolutions}

\begin{table}[h!]
\centering
\renewcommand{\arraystretch}{1.2}
$
\scalebox{.91}{$
	\begin{array}{|c|c|c|}
	\hline
	\text{Group} & \text{Fiber Type} & \text{Crepant Resolution} \\\hline
		\text{SU}(2) & 
		\begin{matrix}
		\text{I}_2^\text{s},     \text{I}^{\text{ns}}_2\\
		 {\text{I}^{\text{ns}}_3},  \text{III} \\   \text{IV}^{\text{ns}} 
		 \end{matrix} & \begin{array}{c} \begin{tikzpicture}
	\node(X0) at (0,0){$X_0$};
	\node(X1) at (2.5,0){$X_1$};
	\draw[big arrow] (X1) -- node[above,midway]{$(x,y,s|e_1)$} (X0);		
 \end{tikzpicture}\end{array} \\\hline
 
		\begin{matrix} \text{SU}(3) \\ \text{USp}(4) \\ \text{G}_2 \end{matrix} & \begin{matrix} \text{I}_3^\text{s} , { \text{IV}^\text{s}} \\ \text{I}_4^\text{ns} \\ \text{I}_0^{* \text{ns}} \end{matrix} & \begin{array}{c} \begin{tikzpicture}
	\node(X0) at (0,0){$X_0$};
	\node(X1) at (2.5,0){$X_1$};
	\node(X2) at (5,0){$X_2$};
	\draw[big arrow] (X1) -- node[above,midway]{$(x,y,s|e_1)$} (X0);	
	\draw[big arrow] (X2) -- node[above,midway]{$(y,e_1|e_2)$} (X1);		
 \end{tikzpicture}\end{array}  \\\hline
	\begin{matrix}	\text{SU}(4) \\ \text{Spin}(7) \end{matrix} &\begin{matrix}  \text{I}_4^\text{s} \\ \text{I}_0^{* \text{ss}} \end{matrix} & \begin{array}{c} \begin{tikzpicture}
	\node(X0) at (0,0){$X_0$};
	\node(X1) at (2.5,0){$X_1$};
	\node(X2) at (5,0){$X_2$};
	\node(X3) at (7.5,0){$X_3$};
	\draw[big arrow] (X1) -- node[above,midway]{$(x,y,s|e_1)$} (X0);	
	\draw[big arrow] (X2) -- node[above,midway]{$(y,e_1|e_2)$} (X1);
	\draw[big arrow] (X3) -- node[above,midway]{$(x,e_2|e_3)$} (X2);		
 \end{tikzpicture}\end{array}  \\\hline
		\text{Spin}(8) &    \text{I}^{*\text{s}}_0 &  \begin{array}{c} \begin{tikzpicture}
	\node(X0) at (0,0){$X_0$};
	\node(X1) at (2.5,0){$X_1$};
	\node(X2) at (5,0){$X_2$};
	\node(X3) at (8,0){$X_3$};
	\node(X4) at (11,0){$X_4$};
	\draw[big arrow] (X1) -- node[above,midway]{$(x,y,s|e_1)$} (X0);	
	\draw[big arrow] (X2) -- node[above,midway]{$(y,e_1|e_2)$} (X1);
	\draw[big arrow] (X3) -- node[above,midway]{$(x-x_i s z,e_2|e_3)$} (X2);		
	\draw[big arrow] (X4) -- node[above,midway]{$(x-x_j s z,e_2|e_4)$} (X3);
 \end{tikzpicture}\end{array} \\\hline
		\text{F}_4 &  \text{IV}^{*\text{ns}} &  \begin{array}{c} \begin{tikzpicture}
	\node(X0) at (0,0){$X_0$};
	\node(X1) at (2.5,0){$X_1$};
	\node(X2) at (5,0){$X_2$};
	\node(X3) at (8,0){$X_3$};
	\node(X4) at (11,0){$X_4$};
	\draw[big arrow] (X1) -- node[above,midway]{$(x,y,s|e_1)$} (X0);	
	\draw[big arrow] (X2) -- node[above,midway]{$(y,e_1|e_2)$} (X1);
	\draw[big arrow] (X3) -- node[above,midway]{$(x,e_2|e_3)$} (X2);		
	\draw[big arrow] (X4) -- node[above,midway]{$(e_3,e_2|e_4)$} (X3);
 \end{tikzpicture}\end{array}\\\hline
 		\text{SU}(5) &  \text{I}_5^{\text{s}} &  \begin{array}{c} \begin{tikzpicture}
	\node(X0) at (0,0){$X_0$};
	\node(X1) at (2.5,0){$X_1$};
	\node(X2) at (5,0){$X_2$};
	\node(X3) at (8,0){$X_3$};
	\node(X4) at (11,0){$X_4$};
	\draw[big arrow] (X1) -- node[above,midway]{$(x,y,s|e_1)$} (X0);	
	\draw[big arrow] (X2) -- node[above,midway]{$(x,y,e_1|e_2)$} (X1);
	\draw[big arrow] (X3) -- node[above,midway]{$(y,e_1|e_3)$} (X2);		
	\draw[big arrow] (X4) -- node[above,midway]{$(y,e_2|e_4)$} (X3);
 \end{tikzpicture}\end{array}\\\hline
 		\text{Spin}(10) &  \text{I}_1^{*\text{s}} &  \begin{array}{c} \begin{tikzpicture}
	\node(X0) at (0,0){$X_0$};
	\node(X1) at (2.5,0){$X_1$};
	\node(X2) at (5,0){$X_2$};
	\node(X3) at (7.5,0){$X_3$};
	\node(X4) at (10,0){$X_4$};
	\node(X5) at (12.5,0){$X_5$};
	\draw[big arrow] (X1) -- node[above,midway]{$(x,y,s|e_1)$} (X0);	
	\draw[big arrow] (X2) -- node[above,midway]{$(y,e_1|e_2)$} (X1);
	\draw[big arrow] (X3) -- node[above,midway]{$(x,e_2|e_3)$} (X2);		
	\draw[big arrow] (X4) -- node[above,midway]{$(y,e_3|e_4)$} (X3);
	\draw[big arrow] (X5) -- node[above,midway]{$(e_2,e_3|e_5)$} (X4);
 \end{tikzpicture}\end{array}\\\hline
 \text{SU}(6)  & \text{I}_6^{\text{s}} &   \begin{array}{c} \begin{tikzpicture}
	\node(X0) at (0,0){$X_0$};
	\node(X1) at (2.5,0){$X_1$};
	\node(X2) at (5,0){$X_2$};
	\node(X3) at (7.5,0){$X_3$};
	\node(X4) at (10,0){$X_4$};
	\node(X5) at (12.5,0){$X_5$};
	\draw[big arrow] (X1) -- node[above,midway]{$(x,y,s|e_1)$} (X0);	
	\draw[big arrow] (X2) -- node[above,midway]{$(y,e_1|e_2)$} (X1);
	\draw[big arrow] (X3) -- node[above,midway]{$(x,e_2|e_3)$} (X2);		
	\draw[big arrow] (X4) -- node[above,midway]{$(y,e_3|e_4)$} (X3);
	\draw[big arrow] (X5) -- node[above,midway]{$(x,e_4|e_5)$} (X4);
 \end{tikzpicture}\end{array} \\\hline
 \text{SU}(7)  & \text{I}_7^{\text{s}} &   \begin{array}{c} \begin{tikzpicture}
	\node(X0) at (0,0){$X_0$};
	\node(X1) at (2.5,0){$X_1$};
	\node(X2) at (5,0){$X_2$};
	\node(X3) at (7.5,0){$X_3$};
	\node(X4) at (10,0){$X_4$};
	\node(X5) at (10,-1){$X_5$};
	\node(X6) at (7.5,-1){$X_6$};
	\draw[big arrow] (X1) -- node[above,midway]{$(x,y,s|e_1)$} (X0);	
	\draw[big arrow] (X2) -- node[above,midway]{$(y,e_1|e_2)$} (X1);
	\draw[big arrow] (X3) -- node[above,midway]{$(x,e_2|e_3)$} (X2);		
	\draw[big arrow] (X4) -- node[above,midway]{$(y,e_3|e_4)$} (X3);
	\draw[big arrow] (X5) -- node[right,midway]{$(x,e_4|e_5)$} (X4);
	\draw[big arrow] (X6) -- node[above,midway]{$(y,e_5|e_6)$} (X5);
 \end{tikzpicture}\end{array} \\\hline
 \text{E}_6  & \text{IV}^{*\text{s}} &   \begin{array}{c} \begin{tikzpicture}
	\node(X0) at (0,0){$X_0$};
	\node(X1) at (2.5,0){$X_1$};
	\node(X2) at (5,0){$X_2$};
	\node(X3) at (7.5,0){$X_3$};
	\node(X4) at (10,0){$X_4$};
	\node(X5) at (10,-1){$X_5$};
	\node(X6) at (7.5,-1){$X_6$};
	\draw[big arrow] (X1) -- node[above,midway]{$(x,y,s|e_1)$} (X0);	
	\draw[big arrow] (X2) -- node[above,midway]{$(y,e_1|e_2)$} (X1);
	\draw[big arrow] (X3) -- node[above,midway]{$(x,e_2|e_3)$} (X2);		
	\draw[big arrow] (X4) -- node[above,midway]{$(e_2,e_3|e_4)$} (X3);
	\draw[big arrow] (X5) -- node[right,midway]{$(y,e_3|e_5)$} (X4);
	\draw[big arrow] (X6) -- node[above,midway]{$(y,e_4|e_6)$} (X5);
 \end{tikzpicture}\end{array} \\\hline
 \text{E}_7 &\text{III}^* &  \begin{array}{c}\begin{tikzpicture}
 	\node(X0) at (0,0){$X_0$};
	\node(X1) at (2.5,0){$X_1$};
	\node(X2) at (5,0){$X_2$};
	\node(X3) at (7.5,0){$X_3$};
	\node(X4) at (10,0){$X_4$};
	\node(X5) at (10,-1){$X_5$};
	\node(X6) at (7.5,-1){$X_6$};
	\node(X7) at (5,-1){$X_7$};
	\draw[big arrow] (X1) -- node[above,midway]{$(x,y,s|e_1)$} (X0);	
	\draw[big arrow] (X2) -- node[above,midway]{$(y,e_1|e_2)$} (X1);
	\draw[big arrow] (X3) -- node[above,midway]{$(x,e_2|e_3)$} (X2);		
	\draw[big arrow] (X4) -- node[above,midway]{$(y,e_3|e_4)$} (X3);
	\draw[big arrow] (X5) -- node[right,midway]{$(e_2,e_3|e_5)$} (X4);
	\draw[big arrow] (X6) -- node[above,sloped,midway]{$(e_2,e_4|e_6)$} (X5);
	\draw[big arrow] (X7) -- node[above,sloped,midway]{$(e_4,e_5|e_7)$} (X6);
 \end{tikzpicture}	\end{array}\\\hline
 \text{E}_8 & \text{II}^*  & \begin{array}{c}\begin{tikzpicture}
 	\node(X0) at (0,0){$X_0$};
	\node(X1) at (2.5,0){$X_1$};
	\node(X2) at (5,0){$X_2$};
	\node(X3) at (7.5,0){$X_3$};
	\node(X4) at (10,0){$X_4$};
	\node(X5) at (10,-1){$X_5$};
	\node(X6) at (7.5,-1){$X_6$};
	\node(X7) at (5,-1){$X_7$};
	\node(X8) at (2.5,-1){$X_8$};
	\draw[big arrow] (X1) -- node[above,midway]{$(x,y,s|e_1)$} (X0);	
	\draw[big arrow] (X2) -- node[above,midway]{$(y,e_1|e_2)$} (X1);
	\draw[big arrow] (X3) -- node[above,midway]{$(x,e_2|e_3)$} (X2);		
	\draw[big arrow] (X4) -- node[above,midway]{$(y,e_3|e_4)$} (X3);
	\draw[big arrow] (X5) -- node[right,midway]{$(e_2,e_3|e_5)$} (X4);
	\draw[big arrow] (X6) -- node[above,midway]{$(e_4,e_5|e_6)$} (X5);
	\draw[big arrow] (X7) -- node[above,midway]{$(e_2,e_4,e_6|e_7)$} (X6);
	\draw[big arrow] (X8) -- node[above,midway]{$(e_4,e_7|e_8)$} (X7);
 \end{tikzpicture}	\end{array} \\\hline
	\end{array}
	$}
	$
\caption{The blowup centers of the crepant resolutions. The variable $s$ is a section of the line bundle $\mathscr{O}_B(S)$. In other words, the zero locus locus $V(s)$ is the divisor $S$ supporting the singular fiber  given in the second column. }
\label{tab:blowupcenters}
\end{table}	
\clearpage 

\subsection{Tables of results}

\begin{table}[h!]
\begin{center}
\renewcommand{\arraystretch}{2}
\scalebox{1}{ $
\begin{array}{|c|c|c|c|}
\hline 
\text{Algebra} & \text{G} &\text{Kodaira} & \varphi_* \Big(c_3(TY)[Y]\Big) \\
\hline
\text{A}_1 & \text{SU}(2) & \displaystyle 
\begin{matrix}
		\text{I}_2^\text{s},     \text{I}^{\text{ns}}_2,  {\text{I}^{\text{ns}}_3}\\
		  \text{III}, \text{IV}^{\text{ns}} 
\end{matrix}
&  12 L (c_1-6 L)+30 L S-6 S^2  \\
\hline
\begin{array}{c}
\text{A}_2 \\
\text{C}_2 \\
\text{G}_2 
\end{array}& \begin{array} {c}
\text{SU}(3) \\
\text{USp}(4) \\
\text{G}_2
\end{array}& \begin{array} {c}
\text{I}^{\text{s}}_3, \ { \text{IV}^{\text{s}}} \\
\text{I}_4^{\text{ns}} \\
\text{I}^{* \text{ns}}_0
\end{array}  & 12L (c_1-6 L)+48 L S-12S^2  \\
\hline
\begin{array}{c}{ 
\text{A}_3} \\
\text{B}_3 \\
\end{array}& \begin{array} {c}
\text{SU}(4) \\
\text{Spin}(7) \\
\end{array} & \begin{array} {c}
\text{I}^{\text{s}}_4 \\
\text{I}^{* \text{ss}}_0
\end{array}      & 12 L (c_1-6 L)+64 L S-20 S^2  \\
\hline
\begin{array}{c}
\text{D}_4 \\
\text{F}_4 
\end{array} & \begin{array} {c}
\text{Spin}(8) \\
\text{F}_4
\end{array} & \begin{array} {c}
\text{I}^{*\text{s}}_0 \\
 \text{IV}^{*\text{ns}}
\end{array}      & 12L (c_1-6 L)+72 L S-24 S^2 \\
\hline 
\text{A}_4 &\text{SU}(5) &\text{I}_5^\text{s} & 12 L (c_1-6 L)+80 L S-30 S^2  \\
\hline 
\text{D}_5 &\text{Spin}(10) &\text{I}_1^{*\text{s}}& 12 L (c_1-6 L)+84 L S-32 S^2  \\
\hline 
\text{A}_5 &\text{SU}(6) &\text{I}_6^\text{s} & 12 L (c_1-6 L)+96 L S - 42 S^2  \\
\hline 
\text{A}_6 &\text{SU}(7) &\text{I}_7^\text{s} & 12 L (c_1-6 L)+112 L S - 56 S^2  \\
\hline 
\text{E}_6 & \text{E}_6 &  \text{IV}^{*\text{s}}  & 12 L (c_1-6 L)+90 L S-36 S^2  \\
\hline 
\text{E}_7 & \text{E}_7 & \text{III}^{*} & 12 L (c_1-6 L)+98 L S-42 S^2  \\
\hline 
\text{E}_8 & \text{E}_8 &  \text{II}^{*}      & 12L (c_1-6 L)+120 L S-60 S^2  \\
\hline 
\end{array} 
$ }
\end{center}
\caption{Chern numbers after pushforwards to the base. The divisor $S$ is the one supporting the reducible Kodaira fiber corresponding to the type of the Lie algebra  $\mathfrak{g}$. By definition, $L=c_1(\mathscr{L})$ and  $c_i$ denotes the $i$th Chern class of the base of the fibration.}
 \label{Table:ChernNumbers}
\end{table}

 \clearpage 
\begin{table}[htb]
\begin{center}
\renewcommand{\arraystretch}{1.5}
\scalebox{.72}{$
\begin{array}{|c|c|c|c|c|c|}
\hline 
\text{Group} & c_1^4 (TY) & c_1^2 (TY) c_2(TY) & c_2^2 (TY) & c_1(TY) c_3(TY) & c_4(TY) \\\hline
\text{SU}(2) & 0 & 12 L (c_1-L)^2 & \begin{array} {c}
24 L (L (6 L-c_1)+c_2) \\
-2 S (7 L-S)^2
\end{array} & \begin{array} {c}
12 L (c_1-6 L) (c_1-L) \\
+6 S (c_1-L) (5 L-S)
\end{array} & \begin{array} {c}
12 L (6 L (6 L-c_1)+c_2) \\
+6 L S (5 c_1-54 L) \\
-6 S^2 (c_1-15 L)-6 S^3 
\end{array} \\
\hline
\begin{array} {c}
\text{SU}(3) \\
\text{USp}(4) \\
\text{G}_2 
\end{array} & 0 & 12 L (c_1-L)^2 & \begin{array} {c}
24 L (L (6 L-c_1)+c_2) \\
-8 S \left(19 L^2-8 L S+S^2\right)
\end{array} & \begin{array} {c}
12 L (c_1-6 L) (c_1-L) \\
+12 S (c_1-L) (4 L-S)
\end{array} & \begin{array} {c} 
12 L (6 L (6 L-c_1)+c_2) \\
+24 L S (2 c_1-21 L) \\
-12 S^2 (c_1-17 L)-24 S^3
\end{array} \\
\hline
\begin{array} {c}
\text{SU}(4) \\
\text{Spin}(7)
\end{array} & 0 & 12 L (c_1-L)^2 & \begin{array} {c}
24 L (L (6 L-c_1)+c_2) \\
-4 S \left(50 L^2-28 L S+5 S^2\right)
\end{array} & \begin{array} {c}
12 L (c_1-L) (c_1-6 L) \\
+4 S (c_1-L) (16 L-5 S)
\end{array} & \begin{array} {c}
12 L (6 L (6 L-c_1)+c_2) \\
+8 L S (8 c_1-83 L) \\
+4 S^2 (89 L-5 c_1)-60 S^3
\end{array} \\
\hline
\begin{array} {c} \text{Spin}(8) \\ \text{F}_4 \end{array} & 0 & 12 L (c_1-L)^2 & \begin{array} {c} 
24 L (L (6 L-c_1)+c_2) \\
-8 S \left(27 L^2-16 L S+3 S^2\right)
\end{array} & \begin{array} {c} 
12 L (c_1-L) (c_1-6 L) \\
+24 S (c_1-L) (3 L-S)
\end{array} & \begin{array} {c}
12 L (6 L (6 L-c_1)+c_2) \\
+72 L S (c_1-10 L) \\
-24 S^2 (c_1-17 L)-72 S^3
\end{array} \\
\hline 
\text{SU}(5) & 0 & 12 L (c_1-L)^2 & \begin{array} {c} 
24 L (L (6 L-c_1)+c_2) \\
-5 S \left(50 L^2-35 L S+8 S^2\right)
\end{array} & \begin{array} {c} 
12 L (c_1-L) (c_1-6 L) \\
+10 S (c_1-L) (8 L-3 S)
\end{array} & \begin{array} {c} 
12 L (6 L (6 L-c_1)+c_2) \\
+10 L S (8 c_1-83 L) \\
+15 S^2 (37 L-2 c_1)-120 S^3 \end{array} \\
\hline 
\text{Spin}(10) & 0 & 12 L (c_1-L)^2 & \begin{array} {c}  
24 L (L (6 L-c_1)+c_2) \\
-4 S \left(63 L^2-44 L S+10 S^2\right)
\end{array} & \begin{array} {c} 
12 L (c_1-L) (c_1-6 L) \\
+4 S (c_1-L) (21 L-8 S)
\end{array} & \begin{array} {c} 
12 L (6 L (6 L-c_1)+c_2) \\
+84 L S (c_1-10 L) \\
-16 S^2 (2 c_1-35 L)-120 S^3
\end{array} \\
\hline 
\text{SU}(6) & 0 & 12 L (c_1-L)^2 & \begin{array} {c} 
24 L (L (6 L-c_1)+c_2) \\
-S (298 L^2-251 L S+70 S^2)
\end{array} & \begin{array} {c} 
12 L (c_1-L) (c_1-6 L) \\
+6 S (c_1-L) (16 L-7 S)
\end{array} & \begin{array} {c} 
12 L (6 L (6 L-c_1)+c_2) \\
+6 L S (16 c_1 - 165 L) \\
-3 S^2 (14 c_1-265 L)-210 S^3
\end{array} \\
\hline 
\text{SU}(7) & 0 & 12 L (c_1-L)^2 & \begin{array} {c} 
24 L (L (6 L-c_1)+c_2) \\
-2 S (174 L^2-171 L S+56 S^2)
\end{array} & \begin{array} {c} 
12 L (c_1-L) (c_1-6 L) \\
+56 S (c_1-L) (2 L-S)
\end{array} & \begin{array} {c} 
12 L (6 L (6 L-c_1)+c_2) \\
+4 L S (28 c_1 - 289 L) \\
-2 S^2 (28 c_1-541 L)-336 S^3
\end{array} \\
\hline 
\text{E}_6 & 0 & 12 L (c_1-L)^2 & \begin{array} {c} 
24 L (L (6 L-c_1)+c_2) \\
-3 S \left(86 L^2-61 L S+14 S^2\right)
\end{array} & \begin{array} {c} 
12 L (c_1-L) (c_1-6 L) \\
+18 S (c_1-L) (5 L-2 S)
\end{array} & \begin{array} {c} 
12 L (6 L (6 L-c_1)+c_2) \\
+18 L S (5 c_1-48 L) \\
-9 S^2 (4 c_1-65 L)-126 S^3
\end{array} \\
\hline 
\text{E}_7 & 0 & 12 L (c_1-L)^2 & \begin{array} {c}  
24 L (L (6 L-c_1)+c_2) \\
-2 S \left(135 L^2-100 L S+24 S^2\right)
\end{array} & \begin{array} {c} 
12 L (c_1-L) (c_1-6 L) \\
+14 S (c_1-L) (7 L-3 S)
\end{array} & \begin{array} {c} 
12 L (6 L (6 L-c_1)+c_2) \\
+2 L S (49 c_1-454 L) \\
-6 S^2 (7 c_1-107 L)-144 S^3
\end{array} \\
\hline 
\text{E}_8 & 0 & 12 L (c_1-L)^2 & \begin{array} {c}  
24 L (L (6 L-c_1)+c_2) \\
-40 S \left(8 L^2-7 L S+2 S^2\right)
\end{array} & \begin{array} {c} 
12 L (c_1-L) (c_1-6 L) \\
+60 S (c_1-L) (2 L-S)
\end{array} & \begin{array} {c} 
12 L (6 L (6 L-c_1)+c_2) \\
+120 L S (c_1-9 L) \\
-60 S^2 (c_1-15 L)-240 S^3
\end{array} \\
\hline 
\end{array}
$}
\end{center}
\caption{Chern numbers of elliptically fibered fourfolds obtained from crepant resolutions of Tate's models. The divisor $S$ is supporting the reducible Kodaira fiber corresponding to the type of the Lie algebra. We abuse notation and omit the degree $\int$ in the entries of the table. By definition, $L=c_1(\mathscr{L})$ and  $c_i$ denotes the $i$th Chern class of the base of the fibration.}
\label{Table.ChernNumbers}
\end{table}
\clearpage

\begin{table}[htb]
\begin{center}
\renewcommand{\arraystretch}{2.1}
\scalebox{.66}{$
\begin{array}{|c|c|c|c|}
\hline 
\text{Group} & \chi_0 & \chi_1 & \chi_2 \\\hline
\text{SU}(2) & \begin{array} {c} 
\frac{1}{12} L \left(c_1^2-3 c_1 L+c_2+2 L^2\right)
\end{array} & \begin{array} {c} 
-\frac{1}{3} L \left(2 c_1^2-54 c_1 L+5 c_2+232 L^2\right)\\
+\frac{1}{2} S^2 (3 c_1-31 L)-\frac{1}{2} L S (15 c_1-113 L)+S^3
\end{array} & \begin{array} {c} 
-\frac{1}{2} L \left(3 c_1^2+71 c_1 L-17 c_2-554 L^2\right)\\
+S^2 (59 L-3 c_1)+L S (15 c_1-211 L)-4 S^3
\end{array} \\ 
\hline
\begin{array} {c} \text{SU}(3) \\ \text{USp}(4) \\ \text{G}_2 \end{array} & \frac{1}{12} L \left(c_1^2-3 c_1 L+c_2+2 L^2\right) & \begin{array} {c} 
-\frac{1}{3} L \left(2 c_1^2-54 c_1 L+5 c_2+232 L^2\right)\\
+S^2 (3 c_1-35 L)-4 L S (3 c_1-22 L)+4 S^3
\end{array} & \begin{array} {c} 
-\frac{1}{2} L \left(3 c_1^2+71 c_1 L-17 c_2-554 L^2\right)\\
-2 S^2 (3 c_1-67 L)+8 L S (3 c_1-41 L)-16 S^3 \end{array} \\ 
\hline
\begin{array} {c} \text{SU}(4) \\ \text{Spin}(7) \end{array} & \begin{array} {c} 
\frac{1}{12} L \left(c_1^2-3 c_1 L+c_2+2 L^2\right)
\end{array} & \begin{array} {c} 
-\frac{1}{3} L \left(2 c_1^2-54 c_1 L+5 c_2+232 L^2\right)\\
+S^2 (5 c_1-61 L)-4 L S (4 c_1-29 L)+10 S^3
\end{array} & \begin{array} {c} 
-\frac{1}{2} L \left(3 c_1^2+71 c_1 L-17 c_2-554 L^2\right)\\
-2 S^2 (5 c_1-117 L)+16 L S (2 c_1-27 L)-40 S^3
\end{array} \\ 
\hline
\begin{array} {c} \text{Spin}(8)\\ \text{F}_4 \end{array} & \begin{array} {c} 
\frac{1}{12} L \left(c_1^2-3 c_1 L+c_2+2 L^2\right)
\end{array} & \begin{array} {c} 
-\frac{1}{3} L \left(2 c_1^2-54 c_1 L+5 c_2+232 L^2\right)\\
+2 S^2 (3 c_1-35 L)-18 L S (c_1-7 L)+12 S^3
\end{array} & \begin{array} {c} 
-\frac{1}{2} L \left(3 c_1^2+71 c_1 L-17 c_2-554 L^2\right)\\
-4 S^2 (3 c_1-67 L)+36 L S (c_1-13 L)-48 S^3
\end{array} \\ 
\hline 
\text{SU}(5) & \begin{array} {c} 
\frac{1}{12} L \left(c_1^2-3 c_1 L+c_2+2 L^2\right)
\end{array} & \begin{array} {c} 
-\frac{1}{3} L \left(2 c_1^2-54 c_1 L+5 c_2+232 L^2\right)\\
+\frac{5}{2} S^2 (3 c_1-38 L)-5 L S (4 c_1-29 L)+20 S^3
\end{array} & \begin{array} {c} 
-\frac{1}{2} L \left(3 c_1^2+71 c_1 L-17 c_2-554 L^2\right)\\
-5 S^2 (3 c_1-73 L)+20 L S (2 c_1-27 L)-80 S^3
\end{array} \\ 
\hline 
\text{Spin}(10) & \begin{array} {c} 
\frac{1}{12} L \left(c_1^2-3 c_1 L+c_2+2 L^2\right)
\end{array} & \begin{array} {c} 
-\frac{1}{3} L \left(2 c_1^2-54 c_1 L+5 c_2+232 L^2\right)\\
+8 S^2 (c_1-12 L)-21 L S (c_1-7 L)+20 S^3
\end{array} & \begin{array} {c} 
-\frac{1}{3} L \left(2 c_1^2-54 c_1 L+5 c_2+232 L^2\right)\\
+8 S^2 (c_1-12 L)-21 L S (c_1-7 L)+20 S^3
\end{array} \\ 
\hline 
\text{SU}(6) & \begin{array} {c} 
\frac{1}{12} L \left(c_1^2-3 c_1 L+c_2+2 L^2\right)
\end{array} & \begin{array} {c} 
-\frac{1}{3} L \left(2 c_1^2-54 c_1 L+5 c_2+232 L^2\right)\\
+\frac{1}{2} S^2 (21 c_1 - 272 L) - L S (24 c_1 - 173 L) + 35 S^3
\end{array} & \begin{array} {c} 
-\frac{1}{2} L \left(3 c_1^2+71 c_1 L-17 c_2-554 L^2\right)\\
+S^2 (523 L - 21 c_1) + 4 L S (12 c_1 - 161 L) - 140 S^3
\end{array} \\ 
\hline 
\text{SU}(7) & \begin{array} {c} 
\frac{1}{12} L \left(c_1^2-3 c_1 L+c_2+2 L^2\right)
\end{array} & \begin{array} {c} 
-\frac{1}{3} L \left(2 c_1^2-54 c_1 L+5 c_2+232 L^2\right)\\
\left(+S^2\right) (14 c_1-185 L)-2 L S (14 c_1-101 L)+56 S^3
\end{array} & \begin{array} {c} 
-\frac{1}{2} L \left(3 c_1^2+71 c_1 L-17 c_2-554 L^2\right)\\
-4 S^2 (7 c_1-178 L)+8 L S (7 c_1-94 L)-224 S^3
\end{array} \\ 
\hline 
\text{E}_6 & \begin{array} {c} 
\frac{1}{12} L \left(c_1^2-3 c_1 L+c_2+2 L^2\right)
\end{array} & \begin{array} {c} 
-\frac{1}{3} L \left(2 c_1^2-54 c_1 L+5 c_2+232 L^2\right)\\
+\frac{3}{2} S^2 (6 c_1-67 L)-\frac{3}{2} L S (15 c_1-101 L)+21 S^3
\end{array} & \begin{array} {c} 
-\frac{1}{2} L \left(3 c_1^2+71 c_1 L-17 c_2-554 L^2\right)\\
-6 S^2 (3 c_1-64 L)+3 L S (15 c_1-187 L)-84 S^3
\end{array} \\  
\hline 
\text{E}_7 & \begin{array} {c} 
\frac{1}{12} L \left(c_1^2-3 c_1 L+c_2+2 L^2\right)
\end{array} & \begin{array} {c} 
-\frac{1}{3} L \left(2 c_1^2-54 c_1 L+5 c_2+232 L^2\right)\\
+\frac{1}{2} S^2 (21 c_1-221 L)-\frac{1}{2} L S (49 c_1-319 L)+24 S^3
\end{array} & \begin{array} {c} 
-\frac{1}{2} L \left(3 c_1^2+71 c_1 L-17 c_2-554 L^2\right)\\
+S^2 (421 L-21 c_1)+L S (49 c_1-589 L)-96 S^3
\end{array} \\ 
\hline 
\text{E}_8 & \begin{array} {c} 
\frac{1}{12} L \left(c_1^2-3 c_1 L+c_2+2 L^2\right)
\end{array} & \begin{array} {c} 
-\frac{1}{3} L \left(2 c_1^2-54 c_1 L+5 c_2+232 L^2\right)\\
+5 S^2 (3 c_1-31 L)+10 L S (19 L-3 c_1)+40 S^3
\end{array} & \begin{array} {c} 
\frac{1}{2} L \left(-3 c_1^2-71 c_1 L+17 c_2+554 L^2\right)\\
+10 S^2 (59 L-3 c_1)+20 L S (3 c_1-35 L)-160 S^3
\end{array} \\ 
\hline 
\end{array}
$}
\end{center}
\caption{Holomorphic genera. The divisor $S$ is the one supporting the reducible Kodaira fiber corresponding to the type of the Lie algebra  $\mathfrak{g}$. To ease the notation, we abuse notation and omit the degree $\int$ in the entries of the table. By definition, $L=c_1(\mathscr{L})$ and  $c_i$ denotes the $i$th Chern class of the base of the fibration. The holomorphic Euler characteristic $\chi_0(Y)$ is equal to $\chi_0(W,\mathscr{O}_W)$ where $W$ is the divisor defined by $\mathscr{L}$ in 
the base (see section \ref{sec:Toddg}).}
\label{Table.HolomorphicEC}
\end{table}

\begin{table}[htb]
\begin{center}
\renewcommand{\arraystretch}{2.3}
\scalebox{.78}{$
\begin{array}{|c|l|l|}
\hline 
\text{Group}  & \hspace{4cm}\int_Y p_2(TY) & \hspace{4cm} \int_Y p_1^2(TY)  \\\hline
\text{SU}(2)  &  
-24 L (-c_1^2+2 c_2+36 L^2) -14 S (7 L-S)^2
 & -48 L (c_1^2-2 c_2-11 L^2)-8 S (7 L-S)^2
 \\
\hline
\text{SU}(3), \ \text{USp}(4) , \  \text{G}_2   &  
-24 L (-c_1^2+2 c_2+36 L^2) -56 S (19 L^2-8 L S+S^2)
 & -48 L (c_1^2-2 c_2-11 L^2)-32 S (19 L^2-8 L S+S^2)
\\
\hline
\text{SU}(4),\ \text{Spin}(7)  &
-24 L (-c_1^2+2 c_2+36 L^2) -28 S (50 L^2-28 L S+5 S^2)
 & -48 L (c_1^2-2 c_2-11 L^2)-16 S (50 L^2-28 L S+5 S^2)
 \\
\hline
 \text{Spin}(8), \ \text{F}_4  &  
-24 L (-c_1^2+2 c_2+36 L^2) -56 S (27 L^2-16 L S+3 S^2)
 & -48 L (c_1^2-2 c_2-11 L^2)-32 S (27 L^2-16 L S+3 S^2)
 \\
\hline 
\text{SU}(5) &
-24 L (-c_1^2+2 c_2+36 L^2) -35 S (50 L^2-35 L S+8 S^2)
 & -48 L (c_1^2-2 c_2-11 L^2)-20 S (50 L^2-35 L S+8 S^2)
  \\
\hline 
 \text{Spin}(10) &-24 L (-c_1^2+2 c_2+36 L^2) -28 S (63 L^2-44 L S+10 S^2)
 &  -48 L (c_1^2-2 c_2-11 L^2)-16 S (63 L^2 - 44 L S + 10 S^2) 
 \\
\hline 
\text{SU}(6) &
-24 L (-c_1^2+2 c_2+36 L^2)-7 S (298 L^2-251 L S+70 S^2)
 & -48 L (c_1^2-2 c_2-11 L^2)-4 S (298 L^2-251 L S+70 S^2)
  \\
\hline 
\text{SU}(7) &
-24 L (-c_1^2+2 c_2+36 L^2)-14 S (174 L^2-171 L S+56 S^2)
 & -48 L (c_1^2-2 c_2-11 L^2)-8 S (174 L^2-171 L S+56 S^2)
  \\
\hline 
\text{E}_6 & 
-24 L (-c_1^2+2 c_2+36 L^2) -21 S (86 L^2-61 L S+14 S^2)
 & -48 L (c_1^2-2 c_2-11 L^2)-12 S (86 L^2-61 L S+14 S^2)\\
\hline 
\text{E}_7 & 
-24 L (-c_1^2+2 c_2+36 L^2) -14 S (135 L^2-100 L S+24 S^2)
 & -48 L (c_1^2-2 c_2-11 L^2)-8 S (135 L^2-100 L S+24 S^2)
  \\
\hline 
\text{E}_8 &  
-24 L (-c_1^2+2 c_2+36 L^2) -280 S (8 L^2-7 L S+2 S^2)
 & -48 L (c_1^2-2 c_2-11 L^2)-160 S (8 L^2-7 L S+2 S^2)
 \\
\hline 
\end{array}
$}
\end{center}
\caption{Pontryagin numbers. The divisor $S$ is the one supporting the reducible Kodaira fiber corresponding to the type of the Lie algebra  $\mathfrak{g}$. To ease the notation, we abuse notation and omit the degree $\int$ in the entries of the table. By definition, $L=c_1(\mathscr{L})$ and  $c_i$ denotes the $i$th Chern class of the base of the fibration. } 
\label{Table.Pontryagin}
\end{table}

\clearpage
\begin{table}[htb]
\begin{center}
\renewcommand{\arraystretch}{1.6}
\scalebox{.9}{$
\begin{array}{|c|c|c|c|c|c|c|}
\hline 
\text{Group} & 192 X_8=\int_Y( p_1^2-4p_2) &45 \sigma=45\int_Y L_2= \int_Y(7p_2-p_1^2) & 5760\int_Y \hat{\text{A}}_2=\int_Y(7 p_1^2-4 p_2) \\\hline
\text{SU}(2) & \begin{array} {c} 
48 L \left(c_1^2-2 c_2-61 L^2\right) \\
+48 S (7 L-S)^2
\end{array} & \begin{array} {c} 
120 L \left(-c_1^2+2 c_2+46 L^2\right) \\
-90 S (7 L-S)^2
\end{array} & \begin{array} {c} 
240 L \left(-c_1^2+2 c_2+L^2\right)
\end{array} \\
\hline
\begin{array} {c} \text{SU}(3) \\ \text{USp}(4) \\  \text{G}_2 \end{array} & \begin{array} {c} 
48 L \left(c_1^2-2 c_2-61 L^2\right) \\
+192S (19 L^2-8 L S+S^2)
\end{array} & \begin{array} {c} 
120L \left(-c_1^2+2 c_2+46 L^2\right) \\
-360 S (19 L^2-8 L S+S^2)
\end{array} & \begin{array} {c} 
240 L \left(-c_1^2+2 c_2+L^2\right)
\end{array} \\
\hline
\begin{array} {c} \text{SU}(4)\\ \text{Spin}(7) \end{array} & \begin{array} {c} 
48 L \left(c_1^2-2 c_2-61 L^2\right) \\
+96 S(50 L^2-28 L S+5 S^2)
\end{array} & \begin{array} {c} 
120 L \left(-c_1^2+2 c_2+46 L^2\right) \\
-180  S(50 L^2 -28 L S+5 S^2)
\end{array} & \begin{array} {c} 
240 L \left(-c_1^2+2 c_2+L^2\right)
\end{array} \\
\hline
\begin{array} {c} \text{Spin}(8) \\ \text{F}_4 \end{array} & \begin{array} {c} 
48 L \left(c_1^2-2 c_2-61 L^2\right) \\
+192 S(27 L^2 -16 L S+3 S^2)
\end{array} & \begin{array} {c} 
120L \left(-c_1^2+2 c_2+46 L^2\right) \\
-360 S (27 L^2-16 L S+3 S^2)
\end{array} & \begin{array} {c} 
240 L \left(-c_1^2+2 c_2+L^2\right)
\end{array} \\
\hline 
\text{SU}(5) & \begin{array} {c} 
48L \left(c_1^2-2 c_2-61 L^2\right) \\
+120 S(50 L^2-35 L S+8 S^2)
\end{array} & \begin{array} {c} 
120 L \left(-c_1^2+2 c_2+46 L^2\right) \\
-225 S (50 L^2 -35 L S+8 S^2)
\end{array} & \begin{array} {c} 
240 L \left(-c_1^2+2 c_2+L^2\right)
\end{array} \\
\hline 
\text{Spin}(10) & \begin{array} {c} 
48L \left(c_1^2-2 c_2-61 L^2\right) \\
+96 S(63 L^2-44 L S+10 S^2)
\end{array} & \begin{array} {c} 
120 L \left(-c_1^2+2 c_2+46 L^2\right) \\
-180 S(63 L^2-44 L S+10 S^2)
\end{array} & \begin{array} {c} 
240 L \left(-c_1^2+2 c_2+L^2\right)
\end{array} \\
\hline 
\text{SU}(6) & \begin{array} {c} 
48L \left(c_1^2-2 c_2-61 L^2\right) \\
+24 S (298 L^2-251 L S+70 S^2)
\end{array} & \begin{array} {c} 
120 L \left(-c_1^2+2 c_2+46 L^2\right) \\
-45 S (298 L^2-251 L S+70 S^2)
\end{array} & \begin{array} {c} 
240 L \left(-c_1^2+2 c_2+L^2\right)
\end{array} \\
\hline 
\text{SU}(7) & \begin{array} {c} 
48L \left(c_1^2-2 c_2-61 L^2\right) \\
+48 S (174 L^2-171 L S+56 S^2)
\end{array} & \begin{array} {c} 
120 L \left(-c_1^2+2 c_2+46 L^2\right) \\
-90 S (174 L^2-171 L S+56 S^2)
\end{array} & \begin{array} {c} 
240 L \left(-c_1^2+2 c_2+L^2\right)
\end{array} \\
\hline 
\text{E}_6 & \begin{array} {c} 
48 L \left(c_1^2-2 c_2-61 L^2\right) \\
+72  S(86 L^2 -61 L S+14 S^2)
\end{array} & \begin{array} {c} 
120 L \left(-c_1^2+2 c_2+46 L^2\right) \\
-135 S(86 L^2 -61 L S+14 S^2)
\end{array} & \begin{array} {c} 
240 L \left(-c_1^2+2 c_2+L^2\right)
\end{array} \\
\hline 
\text{E}_7 & \begin{array} {c} 
48 L \left(c_1^2-2 c_2-61 L^2\right) \\
+48 S(135 L^2 -100 L S+24 S^2)
\end{array} & \begin{array} {c} 
120 L \left(-c_1^2+2 c_2+46 L^2\right) \\
-90 S (135 L^2 -100 L S+24 S^2)
\end{array} & \begin{array} {c} 
240 L \left(-c_1^2+2 c_2+L^2\right)
\end{array} \\
\hline 
\text{E}_8 & \begin{array} {c} 
48 L \left(c_1^2-2 c_2-61 L^2\right) \\
+960 S(8 L^2 -7 L S+2 S^2)
\end{array} & \begin{array} {c} 
120L \left(-c_1^2+2 c_2+46 L^2\right) \\
-1800  S(8 L^2 -7 L S+2 S^2)
\end{array} & \begin{array} {c} 
240 L \left(-c_1^2+2 c_2+L^2\right)
\end{array} \\
\hline 
\end{array}
$}
\end{center}
\caption{Characteristic invariants: the anomaly invariant $X_8$, the signature $\sigma$, and the index of the $\hat{\text{A}}$-genus. The divisor $S$ is the one supporting the reducible Kodaira fiber corresponding to the type of the Lie algebra  $\mathfrak{g}$. To ease the notation, we abuse notation and omit the degree $\int$ in the entries of the table. By definition, $L=c_1(\mathscr{L})$ and  $c_i$ denotes the $i$th Chern class of the base of the fibration.}
\label{Table.Pontryagin2}
\end{table}
\clearpage

\clearpage

\section*{Acknowledgements}
The authors are grateful to  Paolo Aluffi, Patrick Jefferson,  Sungkyung Kang, Chin-Lung Wang,  and Shing-Tung Yau for discussions.  
M.K. would like to thank Simons workshop 2018, Strings 2018, and String-math 2018 for their hospitality. 
M.E. is supported in part by the National Science Foundation (NSF) grant DMS-1701635  ``Elliptic Fibrations and String Theory''.
M.J.K. would like to acknowledge a partial support from NSF grant PHY-1352084. 
 
\appendix

\section{Comments on the literature}\label{AD}
In the present paper, we compute the Chern numbers, Pontryagin numbers, elliptic genera, $L$-genus, Todd-genus, and $X_8(Y)$ for an elliptic fibration  corresponding to a crepant resolution of a Weierstrass model. We call $S$ the divisor supporting the reduced Kodaira fiber. We give the full contribution from the singularity supported on $S$. The smooth Weierstrass case is retrieved by taking $S=0$. We prove in Theorem \ref{Thm:TheInvariance} that the Chern numbers of a fourfold are $K$-equivalence invariant, and in particular, they are  independent on a choice of a crepant resolution. It is important to remember that this is not always the case as shown by Goresky and MacPherson (see Example \ref{GoMAEx1}).
 We also compute the pushforward of  the Pontriagyn class $p_1$ and of $c_2(TY)$ for an elliptic $n$-fold and show that  $\varphi_* c_2(TY)$ and $\varphi_* p_1(TY)$ are independent of the choice of a Kodaira fiber and independent of $S$). 
 In contrast, the Chern number $\int_Y c_2(TY)^2$ and the Pontryagin number $\int_Y p_1^2(TY)$ do depend on $S$.  We point out that the Chern number $\int_Y c_1(TY)^2 c_2(TY)$, the $\hat{A}$-genus and the Todd-genus are independent of $S$ and gives the same value as a smooth Weierstrass model with the same fundamental line bundle $\mathscr{L}$.  We consider the cases of $G=$ SU($n$) for ($n=2,3,4,5,6,7$), USp($4$), Spin($7$), Spin($8$), Spin($10$),  G$_2$, F$_4$, E$_6$, E$_7$, or E$_8$. 

We would  like to comment on \cite{Lawrie:2018jut},  which has some overlaps with the present paper and previous ones \cite{AE1,AE2,EFY}.  Most of the mathematics of  \cite{Lawrie:2018jut} is summarized in appendix A.2  (see also \cite[{Equations (5.17) and (5.18)}]{Lawrie:2018jut}). 
  In  \cite{Lawrie:2018jut}, the authors compute $X_8(Y)$ and the first Pontryagin class $p_1$ for a smooth fourfold that is a Weierstrass model $\varphi: Y_0\to B$ with fundamental line bundle $\mathscr{L}$  and for $G$-models with G$=$ SU($n$), $n=2,3,4,5,6,7$, D$_5$, E$_6$, E$_7$, and E$_8$. 
They also compute the leading terms in $L$ of the Chern numbers of $Y$  (see \cite[Equation (A.26)]{Lawrie:2018jut}). 
Our computation of $X_8(Y)$ results matches those of \cite[{Equations (5.17) and (5.18)}]{Lawrie:2018jut}) for the groups considered. 
 
 The general formula for the pushforward of $p_1$ is Theorem \ref{thm.1.4}:
$$
 \int_Y p_1(TY)\cdot  \varphi^* \beta=
 \int_B(c_1-L)^2 \beta-24\int_B L\cdot  \beta \quad \beta \in A^*(B).$$
  Theorem  \ref{thmP} and the fourth column of Table \ref{Table.Pontryagin2} give the  exact pushforward of $X_8(Y)$ for $G$-models. This answers a question raised by David Morrison (see \cite[Footnote 10]{Lawrie:2018jut}) for all the groups discussed in this paper.

We note that the proofs of \cite{Lawrie:2018jut} rely on two  key relations   \cite[Equation (A.24) and (A.29)]{Lawrie:2018jut}: 
$$
(\star) \quad \quad \varphi_* c(Y)= 12 L (1+\cdots) c(Y),
\qquad(\star\star)\quad\quad \varphi_* Td(Y)=(1-e^L) Td(B),
$$
where  ``$\cdots$'' stands for terms at least linear  in $S$ or at least quadratic in  $L=c_1(\mathscr{L})$.  We would like to comment on the validity of these two equations. 
\footnote{
The proof of Equation ($\star\star$) is incomplete and  the equation has  a sign typo. For example, an infinite number of terms in the right hand side of the Hirzebruch--Riemann--Roch are put to zero without any justification. There are also additional holes. A complete proof of the correct statement 
$$
\varphi_* \rm{Td}(Y)= (1-e^{-L}) c(B),
$$
is available in \cite[Appendix A]{EFY}. }
The relation ($\star$) is true for a smooth Weierstrass model, as proven in \cite{AE1,AE2} in the following closed form, which is presented here in equation \eqref{Eq.AE1}:
$$
\varphi_* c(Y)=\frac{12 L}{1+L}c(B),
$$
where $c(X)=c(TX)\cap [X]$.
In \cite{Lawrie:2018jut}, the authors argue that ($\star$) is also valid for $G$-models if ``$\cdots$'' takes into account terms at least linear in $S$. 
 If the  relation ($\star$) was true, it implies that by taking the degree, the Euler characteristic takes the form 
$$
(\star') \quad \chi(Y)=12\int_BL(1+\cdots) c(B),
$$
which is not the case as seen in Theorem 6.1 of \cite{Euler} and numerous counter examples listed on Table 7 of \cite{Euler}. Already for the case of an SU($2$)-model, we get for a fourfold:
\begin{equation}
\chi(Y)=12\int_Y L(c_2-36 L^2- 6 L c_1) -6\int_Y S \Big( L(54 L-5 c_1)+6S (c_1-15 L)+S^2\Big),
\end{equation}
where $S$ is the divisor over which the generic fiber is of type I$_2$ or III. The obstruction is explicitly the contribution from the singularities and is supported on $S$. 
We also note that ($\star'$) implies that an elliptic fibration has an Euler characteristic that is a multiple of $12$. However, as seen in Klemm-Lian-Roan-Yau, this is not the case even in the special case of elliptically fibered Calabi-Yaus defined in weighted projective spaces \cite{Klemm:1996ts}. 
 The formula \eqref{Eq.Euler} and \eqref{Eq.EFY} in this paper are also counter-examples to ($\star$).

\end{document}